\documentclass[11pt]{article}
\usepackage[utf8]{inputenc}
\usepackage{hyperref}
\usepackage{amsmath}
\usepackage{amsthm}
\usepackage{amssymb}
\usepackage{authblk}
\usepackage{cite}
\usepackage{lipsum}
\usepackage{geometry}
\usepackage{enumerate}
\usepackage{array}
\usepackage{multicol}
\usepackage{multirow}
\usepackage{xcolor}
\newtheorem{theorem}{Theorem}[section]
\numberwithin{theorem}{section}
\newtheorem{lemma}[theorem]{Lemma}
\newtheorem{defi}[theorem]{Definition}

\newtheorem{prop}[theorem]{Proposition}

\DeclareMathOperator{\tr}{Tr}

\newcommand{\F}{\mathbb F}
\newcommand{\Fbn}{\mathbb{F}_{2^n}}

\newcommand{\Fp}{\mathbb{F}_{p}}
\newcommand{\Fpn}{\mathbb{F}_{p^n}}

\newcommand{\Fpnmul}{\mathbb{F}_{p^n}^*}

\newcommand{\Fpnn}{\mathbb{F}_{p^n} \longrightarrow \mathbb{F}_{p^n}}

\begin{document}

    \title{Low $c$-Differential Uniformity of the Swapped Inverse Function in Odd Characteristic}
    \author{Jaeseong Jeong$^1$, Namhun Koo$^2$, Soonhak Kwon$^1$\\
        \small{\texttt{ Email: wotjd012321@naver.com, nhkoo@ewha.ac.kr, shkwon@skku.edu}}\\
\small{$^1$Applied Algebra and Optimization Research Center, Sungkyunkwan University, Suwon, Korea}\\
\small{$^2$Institute of Mathematical Sciences, Ewha Womans University, Seoul, Korea}\\
    }
\date{}

\maketitle
\begin{abstract}
The study of Boolean functions with low
$c$-differential uniformity has become recently an important topic
of research. However, in odd characteristic case, there are not
many results on the ($c$-)differential uniformity of functions
that are not power functions. In this paper, we investigate the
$c$-differential uniformity of the swapped inverse functions in
odd characteristic, and show that their $c$-differential
uniformities are at most 6 except for some special case.


\end{abstract}

\section{Introduction}

In this paper, we mainly focus on the finite fields in odd characteristic. Throughout this paper, let 
\begin{itemize} 
\item $\Fpn$ be the finite field of $p^n$ elements and $\Fpnmul$ and $\Fpnmul=\Fpn \setminus \{0\}$ be the multiplicative group,
where $p$ is an odd prime.
\item $Inv$ be the multiplicative inverse function on $\Fpn$ such that $Inv(x)=x^{-1}$ for all $x\in\Fpnmul$ and
$Inv(0)=0^{-1}=0$.
\end{itemize}

Recently, Ellinsen et. al.\cite{EFR+20} introduced a new concept $c$-differential uniformity, which is useful to estimate resistance against some variants of differential attack\cite{BCJW02}. 

\begin{defi}(\cite{EFR+20}) Let $F : \Fpnn$ be a function and $c\in\Fpn$.\\
(i) We denote the $c$-differential of $F$ by $_cD_a F(x)=F(x+a)-cF(x)$.\\
(ii) Let $a,b\in\Fpn$. We denote $_c\Delta_F(a,b)$ by the number of solutions in $\Fpn$ of $_cD_a F(x)=b$.\\
(iii) The $c$-differential uniformity of $F$ is defined by $_c \Delta_F = \max\{ _c\Delta_F(a,b) : a,b\in \Fpn\text{ and }a\ne 0 \text{ if } c=1 \}$.\\
(iv) $F$ is called a perfect $c$-nonlinear(PcN) function if $_c \Delta_F =1$.\\
(v) $F$ is called a almost perfect $c$-nonlinear(APcN) function if $_c \Delta_F =2$.
\end{defi}

Functions with low $c$-differential uniformity provide good resistance to differential attacks of this type.  Therefore, finding functions with low $c$-differential uniformity has been a good research topic, and many new classes of functions with low $c$-differential uniformity were proposed. Please see \cite{BC21,BCRP21,HPR+21,JKK22b,MRS+21,Sta21b,Sta21d,TZJT21,WZ21,WLZ21,Yan21,ZH21} for some new progress.

In even characteristic, it is known that $Inv\circ (0,1)$ has good cryptographic properties, for example, low differential uniformity\cite{LWY13e}, high nonlinearity\cite{LWY13e}, low boomerang uniformity\cite{LQSL19}, and low differential-linear uniformity\cite{JKK22a}. Very recently, the $c$-differential uniformity and the $c$-boomerang uniformity are investigated in \cite{Sta21b}. 

However, the $c$-differential uniformity of $Inv\circ (0,1)$ in odd characteristic has not been investigated, and the author of \cite{Sta21b} suggested that it would be a good research topic. In this paper, we investigate the $c$-differential uniformity of  the permutation obtained by swapping two points from $Inv$, which can be expressed as $Inv \circ
(\alpha,\beta)$ with $\alpha\ne \beta\in\Fpn$. In \cite{JKK22b}, we show that the $c$-differential uniformity is not preserved under affine equivalence but is preserved under some specific affine equivalence in even characteristic. The next proposition shows that the $c$-differential uniformity is also preserved under this specific affine equivalence in odd characteristic.

\begin{prop}\label{cdu_aff_prop} We call an affine permutation of the form $A(x)=a_1x+a_0$ where $a_0\in \Fpn$ and $a_1\in\Fpnmul$ an affine permutation of degree one. We call two functions $F$ and $F'$ are affine equivalent of degree one if $F=A_1 \circ F' \circ A_2$ where $A_1$ and $A_2$ are affine permutations of degree one. If two functions $F$ and $F'$ defined on $\Fpn$ are affine equivalent of degree one, then $_c\Delta_F={}_c\Delta_{F'}$.
\end{prop}
\begin{proof}
We write $F=A_1 \circ F' \circ A_2$ and $F''=F'\circ A_2$, where $A_1$ and $A_2$ are affine permutations of degree one. $F'$ and $F''$ have the same $c$-differential uniformity in even characteristic\cite{HPR+21}, and the proof in \cite{HPR+21} can be naturally extended for odd characteristic case. It remains to show that $F=A_1 \circ F''$ and $F''$ has the same $c$-differential uniformity. We denote $A_1(x)=a_1x+a_0$ where $a_1\in\Fpnmul$ and $a_0\in\Fpn$. Let $a,b\in\Fpn$. Then $_cD_aF(x)=b$ implies that
\begin{align*}
b&=F(x+a)-cF(x)=(A_1 \circ F'')(x+a)-c(A_1 \circ F'')(x)=A_1 (F''(x+a))-cA_1(F''(x))\\
&=a_1 F''(x+a)+a_0-c(a_1F''(x)+a_0)=a_1(F''(x+a)-cF''(x))+a_0(1-c).
\end{align*}
So we have $_cD_aF''(x)=a_1^{-1}\left(b+a_0(c-1)\right)$ and hence $_c\Delta_F(a,b)={}_c\Delta_{F''}\left(a,a_1^{-1}\left(b+a_0(c-1)\right)\right)$. Since the map from $b$ to $a_1^{-1}\left(b+a_0(c-1)\right)$ is bijective for all $c\in\Fpn$, we obtain $_c\Delta_F={}_c\Delta_{F''}$.
\end{proof}

If $F'=Inv \circ (\alpha,\beta)$, $A(x)=\alpha x$ where $\alpha\in\Fpnmul$ and $\beta\in\Fpn$, then $F=Inv \circ (1,\alpha^{-1}\beta)=A \circ F' \circ A$ and $F'$ have the same $c$-differential uniformity for all $c\in\Fpn$ by Proposition \ref{cdu_aff_prop}. Therefore, in this paper, we investigate the $c$-differential uniformity of $Inv\circ (1,\gamma)$ where $\gamma\in\Fpn\setminus\{1\}$.

For even characteristic case, there are many
non-power functions with low differential uniformity (see
\cite{JKK22b} and references therein). The situation for odd
characteristic case is quite different. That is, though there are
many functions with low differential uniformity (see
\cite{CM97,DMM+03,HRS99,HS97,JZL+13,LRF21,Pot16,ZH21,ZW10,ZW11}
and the references therein), most of them are power functions.
Also, to the best of our knowledge, there is no published article
on the differential uniformity of $Inv\circ (1,\gamma)$ in odd
characteristic. Therefore, we will also study the special case
$c=1$, i.e., the differential uniformity of $Inv\circ
(1,\gamma)$.

The rest of this paper is organized as follows. In section 2, we introduce some preliminaries which are necessary
in our subsequent sections. In section 3, we investigate the $c$-differential uniformity of $Inv\circ (0,1)$ in odd characteristic,  i.e., the case $\gamma=0$. In section 4, we investigate the $c$-differential uniformity of $Inv\circ (1,\gamma)$ with $\gamma\ne 0,1$. In section 5, we give a concluding remark.

\section{Some Previous Results}\label{sec_pre}

In this section, we recall some previous results. First we briefly mention the $c$-differential uniformity of $Inv$ and $Inv\circ (0,1)$ in even characteristic. 
It is known \cite{EFR+20} that $Inv$ is APcN if and only if $\tr(c)=\tr\left(\frac{1}{c}\right)=1$, otherwise $_c\Delta_{Inv}=3$, where $p=2$ with $c\ne 0,1$ and $\tr$ is the field trace. It is also known \cite{Sta21b} that $_c\Delta_{Inv\circ(0,1)}=4$ when $n\ge 4$. Hence we can see that the swapped inverse function has slightly higher $c$-differential uniformity than the inverse function when $p=2$.

For odd characteristic case, if $c\ne 0,1$, then $Inv$ is APcN if and only if $c\in \{4, 4^{-1}\}$ or  both $c^2-4c$ and $1-4c$ are not squares in $\Fpn$, otherwise, $_c\Delta_{Inv}=3$ (see \cite{EFR+20}). The differential uniformity of $Inv$ is as follows\cite{DMM+03,HRS99} (as in usual notation, we use the notation $D_aF(x)={}_1D_aF(x)$ and $\Delta_F(a,b)={}_1\Delta_F(a,b)$, throughout this paper).
\begin{equation*}
\Delta_{Inv} = 
\begin{cases}
2 &\text{ if }p^n\equiv 2\pmod{3},\\
3 &\text{ if }p=3,\\
4 &\text{ if }p^n\equiv 1\pmod{3}.
\end{cases}
\end{equation*}
To summarize, in odd characteristic, one has $2\le \Delta_{Inv}\le 4$ and $2\le {}_c\Delta_{Inv}\le 3$ where $c\ne 0,1$. We will later show that the swapped inverse function has slightly higher ($c$-)differential uniformities than the inverse function in odd characteristic. For the details, see section \ref{sec_inv01} and section \ref{sec_inv1r}.

The quadratic character in $\Fpn$ is defined as follows 
\begin{equation*}
\chi(x) = 
\begin{cases}
1 & \text{ if }x\in\Fpnmul\text{ is a square,} \\
-1 & \text{ if }x\in\Fpnmul\text{ is a nonsquare,}\\
0 & \text{ if }x=0.
\end{cases}
\end{equation*}
It is known that $\chi(x)=x^{\frac{p^n-1}{2}}$. The followings are also well known and useful for our results.
\begin{equation*}
\begin{array}{ll}
\chi(-1)=
\begin{cases}
1 & \text{ if }p^n\equiv 1 \pmod{4}, \\
-1 & \text{ if }p^n\equiv 3 \pmod{4},
\end{cases}
\ & \ 
\chi(5)=
\begin{cases}
1 & \text{ if }p^n\equiv \pm1 \pmod{5}, \\
-1 & \text{ if }p^n\equiv \pm2 \pmod{5},
\end{cases}
\\ & \\
\chi(-3)=
\begin{cases}
1 & \text{ if }p^n\equiv 1 \pmod{3}, \\
-1 & \text{ if }p^n\equiv 2 \pmod{3},
\end{cases} 
\  & \ \chi(2)=
\begin{cases}
1 & \text{ if }p^n\equiv \pm1 \pmod{8}, \\
-1 & \text{ if }p^n\equiv \pm3 \pmod{8}.
\end{cases}
\end{array}
\end{equation*}
We also denote $\sqrt{x}$ by any square root of $x$ in $\Fpn$ if $\chi(x)=1$ throughout this paper. Next we introduce a well-known lemma about the number of solutions of quadratic equations. 

\begin{lemma}\label{quad_lemma_odd}
Let $a_2\in\Fpnmul$ and $a_1,a_0\in\Fpn$ and $D=a_1^2-4a_0a_2$. Then,
\begin{equation*}
\#\{x\in\Fpn : a_2x^2+a_1x+a_0=0\}=
\begin{cases}
2 &\text{ if }\chi(D)=1,\\
1 &\text{ if }\chi(D)=0,\\
0 &\text{ if }\chi(D)=-1.\\
\end{cases}
\end{equation*}
\end{lemma}

The following lemma is a generalization of Lemma 2.6 in \cite{JKK22b} to the odd characteristic case. This lemma is very simple and useful to our results.

\begin{lemma}\label{cdu_symm_lemma} Let $c\ne 0$. Then $_c\Delta_F(a,b)={} _{c^{-1}}\Delta_F(-a,-bc^{-1})$ for all $a\in\Fpn$ and $b\in\Fpn$, and hence $_c\Delta_F={}_{c^{-1}}\Delta_F$.
\end{lemma}
\begin{proof}
Let $a\in\Fpn$ and $b\in\Fpn$. Then, it is clear that $x$ is a solution of $_c D_aF(x)=b$ if and only if $x+a$ is a solution of $_{c^{-1}}D_{-a}(x)=-bc^{-1}$. Hence we obtain $_c\Delta_F(a,b)={} _{c^{-1}}\Delta_F(-a,-bc^{-1})$. Therefore, we get $_c\Delta_F={} _{c^{-1}}\Delta_F$.
\end{proof}

\section{The $c$-differential uniformity of $Inv\circ (0,1)$}\label{sec_inv01}

Let $F=Inv \circ (0,1)$. Since $F$ is permutation, we have $_c\Delta_F(0,b)=1$ for all $b\in\Fpn$ and $c\in\Fpn\setminus\{1\}$, and we also have $_0\Delta_F(a,b)=1$ for all $a,b\in\Fpn$. Hence, throughout this section we assume that
\begin{equation}\label{inv01}
F=Inv\circ (0,1),\ a\ne 0,\ c\ne 0.
\end{equation}
We also denote throughout this section that
$$P=\{x\in\Fpn : F(x)\ne Inv(x)\}=\{0,1\},\ \ P_a=P\cup\{x-a : x\in P\}=\{0,1,-a,1-a\}.$$
It is easy to see that 
$$_c\Delta_F(a,b)=\#\{x\in P_a : {}_cD_aF(x)=b\}+\#\{x\in \Fpn \setminus P_a : {}_c D_aF(x)=b\}.$$
The following lemma characterizes the case $\#\{x\in \Fpn \setminus P_a : D_aF(x)=b\}=2$.

\begin{lemma}\label{panot2_lemma} Under the same assumption as in \eqref{inv01} let $a\in\Fpnmul$ and $b\in\Fpn$. Then $_c D_aF(x)=b$ has at most two solutions in $\Fpn\setminus P_a$. Furthermore, $_c D_aF(x)=b$ has two solutions in $\Fpn\setminus P_a$ if and only if $(ab+c-1)^2-4abc$ is a nonzero square and $b\ne 0$, $(b+c)a+b+c-1\ne0$, $(1-b)a+b+c-1\ne0$.
\end{lemma}

\begin{proof} If $x\in\Fpn\setminus P_a$ then $b={}_cD_aF(x)={}_cD_a Inv(x)$ implies that $bx^2+(ab+c-1)x+ac=0$ which has at most two solutions. Hence $_c D_aF(x)=b$ has at most two solutions in $\Fpn\setminus P_a$.\\
$_cD_aF(x)=b$ has two solutions in $\Fpn\setminus P_a$ implies that $bx^2+(ab+c-1)x+ac=0$ has two solutions which is equivalent to $(ab+c-1)^2-4abc$ is a nonzero square and $b\ne 0$ by Lemma \ref{quad_lemma_odd}. If there is a solution $x_0\in P_a$ of $bx^2+(ab+c-1)x+ac=0$ then $x_0$ cannot be a solution of $_cD_aF(x)=b$. So we require that $bx^2+(ab+c-1)x+ac=0$ has no solutions in $P_a$. We check this by substituting every element in $P_a$ to $bx^2+(ab+c-1)x+ac=0$. 
\begin{itemize}
\item If $x=0$ or $x=-a$ then we have $a=0$ which is a contradiction to $a\ne 0$.
\item If $x=1$ then we have $(b+c)a+b+c-1=0$.
\item If $x=1-a$ then we have $(1-b)a+b+c-1=0$.
\end{itemize}
Thus $bx^2+(ab+c-1)x+ac=0$ has no solutions in $P_a$ if and only if $(b+c)a+b+c-1\ne0$ and $(1-b)a+b+c-1\ne0$.\\
If all conditions in this theorem are satisfied, then $bx^2+(ab+c-1)x+ac=0$ has two solutions in $\Fpn\setminus P_a$. Since all solutions are not contained in $P_a$, they are also solutions of $_cD_aF(x)=b$.
\end{proof}

In the following lemma, we provide a lower bound on $_c\Delta_F$ where $c\ne 0$. Unfortunately, the following lemma shows that $F$ is not APcN for all $c\ne 0$, for every odd prime $p$ and $n>0$.

\begin{lemma}\label{cdu_bound_lemma} Under the same assumption as in \eqref{inv01}, if $p^n>5$ then $_c\Delta_F \ge 3$.
\end{lemma}
\begin{proof} Let $b=a^{-1}-c={}_cD_aF(0)$. Then $_c D_aF(x)=b$ has at least one solution $x=0$ in $P_a$. We investigate required conditions that $_cD_aF(x)=b$ has two solutions in $\Fpn\setminus P_a$ using Lemma \ref{panot2_lemma}. We require that $(ab+c-1)^2-4abc=(c-ca)^2-4c(1-ca)=c^2a^2+2c^2a+c^2-4c=c^2(a+1)^2-4c$ is a nonzero square.  And we have $(b+c)a+b+c-1=a^{-1}\cdot a+a^{-1}-1=a^{-1}$ and $(1-b)a+b+c-1=(c+1-a^{-1})\cdot a+a^{-1}-1=a^{-1}((c+1)a^2-2a+1)$. So we require $a\ne0$ and $(c+1)a^2-2a+1\ne 0$. Since $(c+1)a^2-2a+1= 0$ has no solutions if $(-2)^2-4(c+1)=-4c$ is not a square by Lemma \ref{quad_lemma_odd}, it is enough to check $-c$ is not a square.

\bigskip
\noindent (i) We assume that $-c$ is a square.
\begin{itemize}
\item If $c\ne -4$, then we set $a=-1$. Then  $c^2(a+1)^2-4c=-4c$ is a nonzero square and  $(c+1)a^2-2a+1=c+4\ne0$.
\item If $c=-4$ and $p\ne 3,5,7$, then we choose $a=-\frac{1}{4}=c^{-1}$. Then $c^2(a+1)^2-4c=25$ is a nonzero square and $(c+1)a^2-2a+1=-\frac{3}{16}-\frac{1}{2}+1=\frac{21}{16}\ne 0$.
\item If $c=-4$ and $p=3,5,7$, then $c^2(a+1)^2-4c=16(a+1)^2+16=(a+1)^2+1$ when $p=3,5$ and $c^2(a+1)^2-4c=2\left((a+1)^2+1\right)$ when $p=7$. Since $2$ is a square for every $n$ when $p=7$, it is enough to show that there is $\alpha_0\in\Fpn\setminus\Fp$ such that both $\alpha_0$ and $\alpha_0+1$ are squares. We denote $\Fpn=\Fp[X]/\langle f(X)\rangle$ for some $f(X)$ and $f(g)=0$ for some $g\in \Fpn$. Then $\Fpn$ is partitioned by subsets of the form $\{x+\alpha : \alpha \in \Fp\}$ where $x\in \Fpn\cdot g$. If all such subsets have at most $\frac{p-1}{2}$ squares then there are at most $\frac{p-1}{2p}(p^{n}-p)$ which is a contradiction to the fact that there are $\frac{1}{2}(p^n-p)$ squares in $\Fpn\setminus\Fp$. Hence there is $\{x+\alpha : \alpha \in \Fp\}\ne \Fp$ containing $\frac{p+1}{2}$ squares. If there are at least $\frac{p+1}{2}$ squares in $\{x_0+\alpha : \alpha \in \Fp\}$ for some $x_0 \in \Fpn\cdot g$, then there is $\alpha_0\in \{x_0+\alpha : \alpha \in \Fp\}$ such that both $\alpha_0$ and $\alpha_0+1$ are squares. Then there is $a\in \Fpn$ such that $\alpha_0=(a+1)^2$, and then $(a+1)^2+1=\alpha_0+1$ is also a square. Since $\alpha_0\not \in\Fp$ and hence $a\not\in\Fp$. When $p=3$ we have $(c+1)a^2-2a+1=a+1\ne 0$. When $p=5$ we have $(c+1)a^2-2a+1=2(a+1)(a-2)\ne 0$. When $p=7$ we have $(c+1)a^2-2a+1=4(a+2)(a+1)\ne0$. 
\end{itemize}
(ii) We assume that $-c$ is not a square. As mentioned above, we only require to check that there is $a\in\Fpn$ such that $c^2(a+1)^2-4c$ is a nonzero square.
\begin{itemize}
\item If $-1$ is a square, then one of $c-4$ and $-c^2+4c=-c(c-4)$ is a square. 
\begin{itemize}
\item If $c-4$ is a square then $4(c-4)=4c-16$ is also a square. If we choose $a\in\Fpn$ with $c^2(a+1)^2=4c-16$ then $c^2(a+1)^2-4c=-16$ is a nonzero square.
\item If $-c^2+4c$ is a square then we choose $a\in\Fpn$ with $c^2(a+1)^2=-c^2+4c$. Then $c^2(a+1)^2-4c=-c^2$ is a nonzero square.
\end{itemize}
\item If $-1$ is not a square, then $c$ is a square. Then one of $2c$ and $-2c$ is a square.
\begin{itemize}
\item If $2c$ is a square then $8c$ is a square. If we choose $a\in\Fpn$ with $c^2(a+1)^2=8c$ then $c^2(a+1)^2-4c=4c$ is a nonzero square.
\item If $-2c$ is a square then $-8c$ is a square. In this case, $p=3$ or one of $3$ and $-3$ is a nonzero square.
\begin{itemize}
\item If $p=3$ and $n>1$ then $\Fpn$ is partitioned by $p^{n-1}$ subsets of the form $\{x+\alpha c : \alpha \in \Fp\}$. By similar argument with the case $c=-4$ and $p\in\{3,5,7\}$ above, there is a subset $\{x+\alpha c : \alpha \in \Fp\}\ne\{\alpha c : \alpha \in \Fp\}$ containing $2$ squares. Then there is $\alpha_0\in \{x_0+\alpha : \alpha \in \Fp\}$ such that both $\alpha_0$ and $\alpha_0+c$ are nonzero squares. We choose $a\in\Fpn$ such that $c^2(a+1)^2=\alpha_0+c$, then $c^2(a+1)^2-4c=\alpha_0$ is also a nonzero square.
\item If $3$ is a nonzero square then we choose $a\in\Fpn$ satisfying $c^2(a+1)^2=-2c$. Then $c^2(a+1)^2-4c=-6c$ is also a nonzero square.
\item If $-3$ is a nonzero square then $-12$ is also a square. If we choose $a$ satisfying $c^2(a+1)^2=-8c$ then $c^2(a+1)^2-4c=-12c$ is also a nonzero square.
\end{itemize}
\end{itemize}
\end{itemize}
Therefore, there is $a\in \Fpnmul$ such that $_cD_a(x)=a^{-1}-c$ has two solutions in $\Fbn\setminus P_a$, and hence $_c\Delta_F(a,a^{-1}-c)\ge3$ for all $p^n>7$. When $p^n=7$ we already show that there is $a\in\F_7$ such that $_c\Delta(a,a^{-1}-c)\ge 3$ for all $c\ne -4$. Since $-4^{-1}=5\ne -4$, we have $_5\Delta(a,a^{-1}-c)\ge3$ and hence $_{-4}\Delta\left(-a,c^{-1}(c-a^{-1})\right)={}_5\Delta(a,a^{-1}-c)\ge3$, by Lemma \ref{cdu_symm_lemma}. Therefore, we have $_c\Delta_F\ge3$ for all $p^n>5$, which completes the proof. 
\end{proof}

We obtain $\# P_a <4$ only when $a=\pm1$. In the following lemma we investigate $_c\Delta_F(a,b)$ when $a=\pm1$.

\begin{lemma}\label{cdu_a1_lemma} Under the same assumption as in \eqref{inv01} let $p^n>5$. Then $_c\Delta_F(1,b)\le 3$ for all $b\in\Fpn\setminus \{2^{-1}\}$ and $c\in\Fpnmul\setminus\{-2^{-1}\}$, and  $_c\Delta_F(-1,b)\le 3$ for all $b\in\Fpn\setminus \{1\}$ and $c\in\Fpnmul\setminus\{-2\}$. Furthermore, 
\begin{equation*}
_{-2^{-1}}\Delta_F(1,2^{-1})={}_{-2}\Delta_F(-1,1)=
\begin{cases}
3 &\text{ if }p^n\equiv \pm 3\pmod{8},\\
5 &\text{ if }p^n\equiv \pm 1\pmod{8}.
\end{cases}
\end{equation*}
\end{lemma}
\begin{proof}
If $a=1$ then $P_a=\{0,1,-1\}$ and then we have $_cD_1F(0)=F(1)-cF(0)=-c$, $_cD_1F(1)=F(2)-cF(1)=2^{-1}$ and $_cD_1(-1)=F(0)-cF(-1)=1+c$.
\begin{itemize}
\item $_cD_1F(0)={}_cD_1F(1) \Leftrightarrow -c=2^{-1}$ : Then we have $c=-2^{-1}$. 
\item $_cD_1F(0)={}_cD_1F(-1) \Leftrightarrow -c=1+c$ : Then we have $c=-2^{-1}$.
\item $_cD_1F(1)={}_cD_1F(-1) \Leftrightarrow 2^{-1}=1+c$ : Then $2c+2=1$ and hence we have $c=-2^{-1}$.
\end{itemize}
So if $c=-2^{-1}$ and $b=2^{-1}$ then $_{c}D_1F(x)=b$ has three solutions in $P_a$. If $c\ne -2^{-1}$ or $b\ne 2^{-1}$ then $_cD_1F(x)=b$ has at most one solution, and hence we have $_c\Delta_F(1,b)\le 3$ for all $b\in\Fpn$, by Lemma \ref{panot2_lemma}.
Now assume that $c=-2^{-1}$ and $b=2^{-1}$. Then $_{-2^{-1}}D_1F(x)=2^{-1}$ has three solutions $x=0,\pm1$ in $P_a$. If $a=1$, $b=2^{-1}$ and $c=-2^{-1}$ in Lemma \ref{panot2_lemma} then
\begin{itemize}
\item$(ab+c-1)^2-4abc=(-1)^2-4\cdot\left(-\frac{1}{4}\right)=2$ is a square if and only if $p^n\equiv \pm 1\pmod{8}$. 
\item $(b+c)a+b+c-1=-1 \ne 0$ and $(1-b)a+b+c-1=-2^{-1} \ne 0$
\end{itemize}
Hence $_{-2^{-1}}D_1F(x)=2^{-1}$ has two solutions in $\Fpn\setminus P_a$ if and only if $p^n\equiv \pm 1\pmod{8}$, by Lemma \ref{panot2_lemma}.

By Lemma \ref{cdu_symm_lemma}, $_{c^{-1}}D_{-1}F(x)=-bc^{-1}$ and $_cD_1F(x)=b$ has the same number of solutions. Hence $_{-2}\Delta_F(-1,1)={}_{-2^{-1}}\Delta_F(1,2^{-1})$, and $_cD_{-1}F(x)=b$ has at most three solutions for all $b\in\Fpn\setminus \{1\}$ and $c\in\Fpnmul\setminus\{-2\}$.
\end{proof}

\begin{theorem}\label{cduf_thm} Under the same assumption as in \eqref{inv01} if $p^n>5$ then $3\le {}_c\Delta_F \le 5$. Moreover, $_c\Delta_F=5$ if and only if at least one of the followings are satisfied :\\
(i) $p^n\equiv \pm1\pmod{8}$ and $c\in\{-2,-2^{-1}\}$\\
(ii) $p^n\equiv 1\pmod{3}$ and one of the followings holds :
\begin{itemize}
\item $c\in\left\{\sqrt{-3}, \frac{1}{\sqrt{-3}}\right\}$, $\sqrt{-3}\ne5$, and $\frac{-3+5\sqrt{-3}}{6}$ is a nonzero square.
\item $c\in\left\{-\sqrt{-3}, -\frac{1}{\sqrt{-3}}\right\}$, $\sqrt{-3}\ne-5$, and $\frac{-3-5\sqrt{-3}}{6}$ is a nonzero square.
\end{itemize}
(iii) $p^n\equiv \pm1\pmod{5}$ and one of the followings holds :
\begin{itemize}
\item $c^2+4c-1=0$, and $7-30c$ is a nonzero square.
\item $c^2-4c-1=0$, and $7-30c^{-1}$ is a nonzero square.
\end{itemize}
\end{theorem}

\begin{proof} The case when $a=\pm1$ is covered by Lemma \ref{cdu_a1_lemma} and we have (i). Hence it remains to show (ii) and (iii) when $a\ne \pm1$. We can see that
\begin{equation}\label{cdaf_pa_list}
\begin{array}{ll}
_cD_aF(0)=F(a)-cF(0)=a^{-1}-c, & {}_cD_aF(1)=F(1+a)-cF(1)=(1+a)^{-1},\\
_cD_aF(-a)=F(0)-cF(-a)=1+ca^{-1}, & {}_cD_aF(1-a)=F(1)-cF(1-a)=c(a-1)^{-1}.
\end{array}
\end{equation}
We investigate all possible cases that $_cD_aF(x)=b$ has at least two solutions in $P_a$.
\begin{equation*}
\begin{array}{ll}
 _c D_aF(0)={}_c D_aF(1)  \Leftrightarrow ca^2+ca-1=0,&
 _c D_aF(0)={}_c D_aF(-a)  \Leftrightarrow a=\frac{1-c}{c+1}\text{ when }c\ne \pm1,\\
 _c D_aF(0)={}_c D_aF(1-a)  \Leftrightarrow ca^2-a+1=0,&
 _c D_aF(1)={}_c D_aF(-a)  \Leftrightarrow a^2+ca+c=0,\\
 _c D_aF(1)={}_c D_aF(1-a)  \Leftrightarrow a=\frac{c+1}{1-c}\text{ when }c\ne \pm1,&
 _c D_aF(-a)={}_c D_aF(1-a)  \Leftrightarrow a^2-a-c=0.
\end{array}
\end{equation*}
If $_cD_aF(0)={}_cD_aF(-a)={}_cD_aF(1)={}_cD_aF(1-a)$ then $a=\frac{1-c}{c+1}=\frac{c+1}{1-c}$ and hence we have $c=0$, a contradiction. Hence $_cD_aF(x)=b$ has at most three solutions in $P_a$. We investigate all possible cases that  $_cD_aF(x)=b$ has three solutions in $P_a$.
\begin{itemize}
\item If $_c D_aF(0)={}_c D_aF(1)={}_c D_aF(-a)$ then $a=\frac{1-c}{c+1}$ is a solution of two equations $ca^2+ca-1=0$ and $a^2+ca+c=0$. Then we have $0=ca^2+ca-1=c\left(\frac{1-c}{c+1}\right)^2+c\cdot \frac{1-c}{c+1}-1=-\frac{3c^2+1}{(c+1)^2}$ and $0=a^2+ca+c=\left(\frac{1-c}{c+1}\right)^2+c\cdot \frac{1-c}{c+1}+c=\frac{3c^2+1}{(c+1)^2}$. Such $c$ exists when $-3$ is a nonzero square, if and only if $p^n\equiv 1\pmod{3}$, and then $c=\pm\frac{1}{\sqrt{-3}}$. Then $a=\frac{\sqrt{-3}\mp1}{\sqrt{-3}\pm1}=\frac{1\pm\sqrt{-3}}{2}$ and $b={}_c D_aF(0)=a^{-1}-c=\frac{2}{1\pm\sqrt{-3}}\mp\frac{1}{\sqrt{-3}}=\frac{3\mp\sqrt{-3}}{6}$. Since 
$(ab+c-1)^2-4abc=\left(\frac{1\pm\sqrt{-3}}{2}\cdot\frac{3\mp\sqrt{-3}}{6}\pm\frac{1}{\sqrt{-3}}-1 \right)^2 \mp 4\cdot \frac{1\pm\sqrt{-3}}{2}\cdot\frac{3\mp\sqrt{-3}}{6}\cdot \frac{1}{\sqrt{-3}}
=\left(\frac{-3\mp\sqrt{-3}}{6}\right)^2+\frac{-2\pm 2\sqrt{-3}}{3}=\frac{-3\pm5\sqrt{-3}}{6}$. Using $b+c=\frac{3\mp\sqrt{-3}}{6}\mp \frac{\sqrt{-3}}{3}=\frac{1\mp\sqrt{-3}}{2}$, we have $a(b+c)+b+c+1=\frac{1\pm \sqrt{-3}}{2}\cdot\frac{1\mp\sqrt{-3}}{2}+\frac{1\mp\sqrt{-3}}{2}+1=\frac{5\mp \sqrt{-3}}{2}$ and $(1-b)a+b+c-1=\frac{3\pm\sqrt{-3}}{6}\cdot \frac{1\pm \sqrt{-3}}{2}+\frac{1\mp\sqrt{-3}}{2}-1=\frac{-3\mp \sqrt{-3}}{6}\ne 0$. So, we have $_c\Delta_F(a,b)=5$ if and only if $p^n\equiv 1\pmod{3}$, $\frac{5\mp \sqrt{-3}}{2}\ne 0$, and $\frac{-3\pm5\sqrt{-3}}{6}$ is a square. Note that all double signs are in same order in above computations this case.
\item If $_c D_aF(0)={}_c D_aF(1)={}_c D_aF(1-a)$ then $a=\frac{c+1}{1-c}$ is a solution of two equations $ca^2+ca-1=0$ and $ca^2-a+1=0$. Then we have $0=c\cdot \left(\frac{c+1}{1-c}\right)^2+c\cdot \frac{c+1}{1-c}-1=\frac{c^2+4c-1}{(1-c)^2}$ and $0=c\cdot \left(\frac{c+1}{1-c}\right)^2-\cdot \frac{c+1}{1-c}+1=\frac{c(c^2+4c-1)}{(1-c)^2}$, and hence $c^2+4c-1=0$. Such $c$ exists if $(-4)^2-4\cdot(-1)=20=2^2\cdot 5$ is a square, if and only if $p=5$ or $p^n\equiv \pm 1\pmod{5}$. We set $b={}_c D_aF(0)=a^{-1}-c=\frac{1-c}{c+1}-c = \frac{1-2c-c^2}{c+1}$. Using the proof of Lemma \ref{cdu_bound_lemma}, we have $(ab+c-1)^2-4abc=c^2(a+1)^2-4c=c^2\cdot \left(\frac{2}{1-c}\right)^2-4c=-\frac{4(c^3-3c^2+c)}{(c-1)^2}=-\frac{4((c-7)(c^2+4c-1)+30c-7)}{(c-1)^2}=-\frac{4(30c-7)}{(c-1)^2}$ and $(b+c)a+b+c+1=a^{-1}\ne 0$ and $(1-b)a+b+c-1=a^{-1}((c+1)a^2-2a+1)=a^{-1}\left((c+1)\cdot \left(\frac{c+1}{1-c}\right)^2-2\cdot \frac{c+1}{1-c}+1\right)=\frac{c(c^2+6c+1)}{a(c-1)^2}=\frac{2c(c+1)}{a(c-1)^2}\ne 0$ since $c\ne -1$. So, we have $_c\Delta_F(a,b)=5$ if and only if $p^n\equiv 0,\pm 1\pmod{5}$ and $7-30c$ is a nonzero square. If $p=5$ then we have $c=-2$, and $7-30c=2$ is a square if and only if $p^n\equiv \pm1 \pmod{8}$ and hence this case is already covered by (i).
\item If $_c D_aF(0)={}_c D_aF(-a)={}_c D_aF(1-a)$ then we have $_{c^{-1}} D_{-a}F(0)={}_{c^{-1}} D_{-a}F(1)={}_{c^{-1}} D_{-a}F(-(-a))$ by Lemma \ref{cdu_symm_lemma}. We set $c^{-1}=\pm\frac{1}{\sqrt{-3}}$ and $-a=\frac{1-c^{-1}}{c^{-1}+1}=\frac{1\pm\sqrt{-3}}{2}$ or equivalently $c=\pm\sqrt{-3}$ and $a=-\frac{1\pm\sqrt{-3}}{2}$, and $b=(-a)^{-1}-c^{-1}=\frac{3\mp\sqrt{-3}}{6}$. By the above observation and Lemma \ref{cdu_symm_lemma}, $5={}_{c^{-1}}\Delta_F(-a,-bc^{-1})={}_c\Delta_F(a,b)$ if and only if $p^n\equiv 1\pmod{3}$, $\frac{5\mp \sqrt{-3}}{2}\ne 0$, and $\frac{-3\pm5\sqrt{-3}}{6}$ is a square.
\item If $_c D_aF(1)={}_c D_aF(-a)={}_c D_aF(1-a)$ then we have $_{c^{-1}} D_{-a}F(0)={}_{c^{-1}} D_{-a}F(1)={}_{c^{-1}} D_{-a}F(1-(-a))$ by Lemma \ref{cdu_symm_lemma}. We assume that $(c^{-1})^2+4c^{-1}-1=0$ and $-a=\frac{c^{-1}+1}{1-c^{-1}}$ or equivalently $c^2-4c-1=0$ and $a=\frac{1+c}{1-c}$, and $b=-a^{-1}-c^{-1}=\frac{1-2c^{-1}-(c^{-1})^2}{1+c^{-1}}$. By the above observation and Lemma \ref{cdu_symm_lemma}, $5={}_{c^{-1}}\Delta_F(-a,b)={}_c\Delta_F(a,b)$ if and only if $p\equiv 0,\pm 1 \pmod{5}$ and $7-30c^{-1}$ is a nonzero square. Similarly with the above case, the case $p=5$ can be covered by (i).
\end{itemize}
Conversely, suppose each condition in (ii) and (iii) is satisfied. Using the same $a$ and $b$ in previous analysis, we can see that $_cD_aF(x)=b$ has three solutions in $P_a$. By Lemma \ref{panot2_lemma}, the given element is a nonzero square, then $_cD_aF(x)=b$ has two solutions in $\Fpn\setminus P_a$. Hence we have $_c\Delta_F(a,b)=5$ which completes the proof.
\end{proof}

Next we further investigate the differential uniformity of $F$, the special case $c=1$.

\begin{theorem}\label{duf_thm} Let $F=Inv\circ(0,1)$ and $p^n >5$. Then $3 \le \Delta_F \le 5$. Furthermore,\\
(i) $\Delta_F=5$ if and only if $p=3$ and $n$ is even.\\
(ii) $\Delta_F=4$ if and only if at least one of the followings is satisfied :
\begin{itemize}
\item $p=5$ and $n\equiv 0\pmod{4}$.
\item $p^n\equiv \pm1 \pmod{5}$ with $p\ne 3$, and at least one of $\frac{-5+\sqrt{5}}{2}$ and $\frac{-5-\sqrt{5}}{2}$ is a square.
\item $p^n\equiv 1 \pmod{3}$ with $p\ne 5$, and at least one of $\frac{-5+3\sqrt{-3}}{2}$ and $\frac{-5-3\sqrt{-3}}{2}$ is a nonzero square.
\end{itemize}
\end{theorem}

\begin{proof} By Theorem \ref{cduf_thm}, $3 \le \Delta_F \le 5$ is straightforward. 

\bigskip
\noindent (i) We first show that $D_aF(x)=b$ has at most two solutions in $\Fpn\setminus P_a$ and hence $\Delta_F(a,b)\le 4$, when $a\in\Fpn\setminus\{0,\pm1\}$ and $b\in\Fpn$. We have the following to substitute $c=1$ in \eqref{cdaf_pa_list}
\begin{equation*}
\begin{array}{ll}
D_aF(0)=F(a)-F(0)=a^{-1}-1, & D_aF(1)=F(1+a)-F(1)=(1+a)^{-1},\\
D_aF(-a)=F(0)-F(-a)=1+a^{-1}, & D_aF(1-a)=F(1)-F(1-a)=(a-1)^{-1}.
\end{array}
\end{equation*}
We can see that $D_aF(x-a)=F(x)-F(x-a)=-D_{-a}(x)$. 
If $D_aF(0)=D_aF(-a)$ we have $1=-1$, a contradiction, and hence we have $D_aF(0)\ne D_aF(-a)$. Similarly we have  $D_aF(1)\ne D_aF(1-a)$. Hence $D_aF(u)\ne D_aF(u-a)$ for $u\in P$. If there is $b\in\Fpn$ such that $D_aF(x)=b$ has at least three solutions in $P_a$, then there is $u\in P_a$ such that both $u$ and $u-a$ are solutions of $D_aF(x)=b$, which is a contradiction. Hence $D_aF(x)=b$ has at most two solutions in $P_a$.

So we have $\Delta_F=5$ if and only if $\Delta_F(1,b)=5$ for some $b\in\Fpn$. We have $c=1=-2^{-1}$ which is equivalent to $1=(-2)\cdot (-2)^{-1}=(-2)\cdot 1=-2$ if and only if $p=3$. Then $3^n\equiv \pm 1\pmod{8}$ if and only if $n$ is even. Hence by Lemma \ref{cdu_a1_lemma}, $\Delta_F=5$ if and only if $p=3$ and $n$ is even.

\bigskip
\noindent (ii) We investigate all possible cases that $D_aF(x)=b$ has at least two solutions in $P_a$ applying $D_aF(u)\ne D_aF(u-a)$ for $u\in P$. Since $D_aF(x-a)=F(x)-F(x-a)=-D_{-a}(x)$, we have $D_aF(-a)=D_aF(u-a)$ if and only if $D_{-a}F(0)=D_{-a}F(u)$ for $u\in\{1,1-a\}$. Hence it is enough to investigate the following cases.
\begin{itemize}
\item $D_aF(0)=D_aF(1)$ implies $a^2+a-1=0$. $D=(-1)^2-4\cdot(-1)=5$ is a nonzero square if and only if $p\equiv \pm 1 \pmod{5}$.
\item $D_aF(0)=D_aF(1-a)$ implies $a^2-a+1=0$. $D=(-1)^2-4\cdot 1 =-3$ is a nonzero square if and only if $p\equiv 1 \pmod{3}$.
\end{itemize}
We set $b=D_aF(0)=a^{-1}-1$. Applying the proof of Lemma \ref{cdu_bound_lemma} we have $(b+c)a+b+c-1=a^{-1}\ne0$ and $(1-b)a+b+c-1=a^{-1}(2a^2-2a+1)$. 
\begin{itemize}
\item If $D_aF(0)=D_aF(1)$ then $a^2=-a+1$. Hence $0=2a^2-2a+1=2(-a+1)-2a+1=-4a+3$ implies $a=\frac{3}{4}$. Then $0=a^2+a-1=\left(\frac{3}{4}\right)^2+\frac{3}{4}-1=\frac{5}{16}$ requires $p=5$.
\item If $D_aF(0)=D_aF(1-a)$ then $a^2-a=-1$. Hence we have $2a^2-2a+1=2(a^2-a)+1=-1\ne0$.
\end{itemize}

\noindent 1. If $p=3$ then it remains to show that $\Delta_F(a,b)\le 3$ for all $a\in \Fpn\setminus\{0,\pm1\}$ and $b\in\Fpn$ when $n$ is odd, by Lemma \ref{cdu_a1_lemma}.  Then $a^2-a+1=0$ has one solution $a=-1$, contradicts to the assumption. Furthermore we have $3^n\equiv \pm 2 \pmod{5}$ if $n$ is odd. So $D_aF(x)=b$ has at most one solution in $P_a$ when $a\in \Fpn\setminus\{0,1\}$, and hence we have $\Delta_F(a,b)\le 3$ for all $a\in \Fpn\setminus\{0,1\}$ and $b\in\Fpn$.

\bigskip
\noindent 2. If $p=5$ then $a^2+a-1=0$ has one solution. It is easy to see that $a=2$ is the only solution of $a^2+a-1=0$ and then $b=D_2F(0)=2^{-1}-1=2$. As mentioned above, $D_2 F(x)=2$ has two solutions $x=0,1$ in $P_a$. And then we have $(ab+c-1)^2-4abc=0$ and hence $bx^2+(ab+c-1)x+ac=0$ has exactly one solution, but by the above observation $x=1-a=-1\in P_a$ is the solution of $bx^2+(ab+c-1)x+ac=0$ and hence it cannot be a solution of $D_2 F(x)=2$. Hence $D_2 F(x)=2$ has no solutions in $\Fpn\setminus P_a$. So we have $\Delta_F(2,2)= 2$. \\
If $n$ is odd then $5^n\equiv 2\pmod{3}$ and hence $D_aF(x)=b$ has at most one solution in $P_a$. By Lemma \ref{cdu_a1_lemma} we have $\Delta_F=3$. If $n$ is even then $5^n\equiv 1 \pmod{3}$ so $a^2+a+1=0$ has two solutions. Then $a^3=-1$ and hence $a^{-1}=-a^2$. So we have $b=D_aF(0)=a^{-1}-1=-a^2-1=-a$ then $(ab+c-1)^2-4abc=a^4+4a^2=-a+4(a-1)=3a+1$. Note that $(3a+1)^2=9a^2+6a+1=-a^2+a+1=2$ and hence $3a+1$ is a square root of $2$. Since the multiplicative order of $2$ in modulo $5$ is $4$, we have $1=\chi(3a+1)=(3a+1)^{\frac{5^n-1}{2}}=2^{\frac{5^n-1}{4}}$ if and only if $4\mid \frac{5^n-1}{4}$ if and only if $5^n\equiv 1 \pmod{16}$. Since the multiplicative order of $5$ in modulo $16$ is $4$, we have $5^n\equiv 1 \pmod{16}$ if and only if $n\equiv 0 \pmod{4}$.  Hence, we have $\Delta_F=4$ if and only if $n\equiv 0\pmod{4}$. 

\bigskip
\noindent 3. Now we investigate the case $p\ne 3,5$.
\begin{itemize}
\item If $p^n\equiv \pm1\pmod{5}$, then $5$ is a square and hence $D_aF(0)=D_aF(1)$ where $a^2+a-1=0$. Then $a=\frac{-1\pm\sqrt{5}}{2}$ and $b=a^{-1}-1$ and $(ab+c-1)^2-4abc=a^2+2a-3=a^2+a-1+a-2=a-2=\frac{-5\pm\sqrt{5}}{2}$. Hence by Lemma \ref{panot2_lemma} we have $\Delta_F(a,b)=4$ if and only if at least one of $\frac{-5+\sqrt{5}}{2}$ and $\frac{-5-\sqrt{5}}{2}$ is a square. Note that $\frac{-5\pm\sqrt{5}}{2}\ne0$, since if $\frac{-5\pm\sqrt{5}}{2}=0$ then $x=\pm\sqrt{5}$ is a solution of $x^2-x=0$ which implies $1=5$ or $0=5$, which is a contradiction to $p\ne 5$.
\item If $p^n\equiv 1\pmod{3}$ then $-3$ is a square and hence $D_aF(0)=D_aF(1-a)$ where $a^2-a+1=0$. Then $a=\frac{1	\pm\sqrt{-3}}{2}$ and $b=a^{-1}-1$ and $(ab+c-1)^2-4abc=a^2+2a-3=a^2-a+1+3a-4=3a-4=\frac{-5\pm3\sqrt{-3}}{2}$. Hence by Lemma \ref{panot2_lemma} we have $\Delta_F(a,b)=4$ if and only if at least one of  $\frac{-5+3\sqrt{-3}}{2}$ and $\frac{-5-3\sqrt{-3}}{2}$ is a nonzero square.
\end{itemize}
Conversely, suppose each condition in (ii) is satisfied. Using the same $a$ and $b$ in previous analysis, we can see that $_cD_aF(x)=b$ has two solutions in $P_a$. By Lemma \ref{panot2_lemma}, the given element in each case is a nonzero square, then $D_aF(x)=b$ has two solutions in $\Fpn\setminus P_a$. Hence we have $\Delta_F(a,b)=4$ which completes the proof.
\end{proof}

\section{The $c$-differential uniformity of $Inv\circ (1,\gamma)$}\label{sec_inv1r}

Let $F=Inv \circ (1,\gamma)$ where $\gamma\in\Fpn\setminus \{0,1\}$. Similarly with section \ref{sec_inv01} since $F$ is permutation, we have $_c\Delta_F(0,b)=1$ for all $b\in\Fpn$ and $c\in\Fpn\setminus\{1\}$, and we also have $_0\Delta_F(a,b)=1$ for all $a,b\in\Fpn$. Hence, throughout this section we assume that
\begin{equation}\label{inv1r}
F=Inv\circ (1,\gamma),\ \gamma\in \Fpn \setminus \{0,1\}, \ a\ne 0,\ c\ne 0
\end{equation}
We also denote throughout this section that
$$P=\{x\in\Fpn : F(x)\ne Inv(x)\}\cup\{0\}=\{0,1,\gamma \},\ \ P_a=P\cup\{x-a : x\in P\}=\{0,1,\gamma,-a,1-a,\gamma-a\}.$$
It is easy to see that 
$$_c\Delta_F(a,b)=\#\{x\in P_a : {}_cD_aF(x)=b\}+\#\{x\in \Fpn \setminus P_a : {}_c D_aF(x)=b\}.$$
The following lemma characterizes the case $\#\{x\in \Fpn \setminus P_a : D_aF(x)=b\}=2$.

\begin{lemma}\label{panot2_lemma_1r} Under the same assumption as in \eqref{inv1r} let $a\in\Fpnmul$ and $b\in\Fpn$. Then $_c D_aF(x)=b$ has at most two solutions in $\Fpn\setminus P_a$. Furthermore, $_c D_aF(x)=b$ has two solutions in $\Fpn\setminus P_a$ if and only if $(ab+c-1)^2-4abc$ is a nonzero square and $b\ne 0$, $(b+c)a+b+c-1\ne0$, $(1-b)a+b+c-1\ne0$, $(b\gamma+c)a+\gamma(b\gamma+c-1) \ne 0$ and $(1-b\gamma)a+\gamma(b\gamma+c-1)\ne0$.
\end{lemma}

\begin{proof} Similarly with Lemma \ref{panot2_lemma}, since $F(x)=Inv(x)$ for all $x\in\Fbn\setminus P_a$, $_c D_aF(x)=b$ has at most two solutions in $\Fbn\setminus P_a$, and has two solutions if and only if $(ab+c-1)^2-4abc$ is a nonzero square and $b\ne 0$. We require that $bx^2+(ab+c-1)x+ac=0$ has no solutions in $P_a$. We check this by substituting every element in $P_a$ to $bx^2+(ab+c-1)x+ac=0$. 
\begin{itemize}
\item If $x=0$ or $x=-a$ then we have $a=0$ which is a contradiction to $a\ne 0$.
\item If $x=1$ then we have $(b+c)a+b+c-1=0$.
\item If $x=1-a$ then we have $(1-b)a+b+c-1=0$.
\item If $x=\gamma$ then we have $(b\gamma+c)a+\gamma(b\gamma+c-1)=0$.
\item If $x=\gamma-a$ then we have $(1-b\gamma)a+\gamma(b\gamma+c-1)=0$.
\end{itemize}
Thus $bx^2+(ab+c-1)x+ac=0$ has no solutions in $P_a$ if and only if $(b+c)a+b+c-1\ne0$, $(1-b)a+b+c-1\ne0$, $(b\gamma+c)a+\gamma(b\gamma+c-1)\ne 0$ and $(1-b\gamma)a+\gamma(b\gamma+c-1)\ne0$.\\
Similarly with Lemma \ref{panot2_lemma}, if all conditions in this theorem are satisfied, then $_cD_aF(x)=b$ has two solutions in $\Fpn\setminus P_a$.
\end{proof}

The following lemma shows symmetry of the $c$-differential uniformity with respect to $\gamma$, and it is useful for our results.

\begin{lemma}\label{cdu_symm_lemma_1r} Let $\gamma\ne0,1$ and $F=Inv\circ(1,\gamma)$ and $F'=Inv\circ (1,\gamma^{-1})$. Then $_c\Delta_{F'}(a,b)={}_c\Delta_F(\gamma a,\gamma^{-1}b)$ for all $a,b\in\Fpn$ and hence $_c\Delta_F={}_c\Delta_{F'}$ for all $c\in\Fpn$.
\end{lemma}
\begin{proof} Let $A(x)=\gamma x$. Then it is not difficult to see that $F'= A \circ F\circ A$. Then $_cD_aF'(x)=b$ implies
\begin{align*} 
b&={}_cD_aF'(x)=(A\circ F\circ A)(x+a)-c(A\circ F\circ A)(x)=(A\circ F)(\gamma x+\gamma a)-c(A\circ F)(\gamma x)\\
&=\gamma\left(F(\gamma x +\gamma a)-F(\gamma x)\right)=\gamma(F(X+\gamma a)-cF(X))=\gamma \cdot {}_cD_{\gamma a}F(X)
\end{align*}
if and only if $\gamma^{-1} b ={}_cD_{\gamma a}F(X)$, where $X=\gamma x$ in above.
\end{proof}

We can see that $\#P_a <6$ if and only if $a\in\{\pm 1, \pm\gamma, \pm (\gamma-1)\}$. We investigate the number of solutions of $_cD_aF(x)=b$ in $P_a$ for this case.

\begin{lemma}\label{panot6_lemma}
Let $a\in\{\pm 1, \pm\gamma, \pm (\gamma-1)\}$. Then\\
(i) $_cD_aF(x)=b$ has five solutions in $P_a$ if and only if $p=5$, $\gamma=-1$, $c=1$ and $(a,b)\in\{(2,3),(3,2)\}$.\\
(ii) $_cD_aF(x)=b$ has four solutions in $P_a$ if and only if $p=5$, $\gamma=-1$, $c=1$ and $(a,b)\in\{(1,-1),(-1,1)\}$.\\
Otherwise, $_cD_aF(x)=b$ has at most three solutions in $P_a$.
\end{lemma}
\begin{proof}
We investigate all possible cases, and details are in Appendix \ref{proof_panot6_lemma}.
\end{proof}

In Lemma \ref{cdu_bound_lemma} we show that the $c$-differential uniformity of $F$ is lower bounded by $3$ for all $c\in \Fpnmul$ when $\gamma=0$. But the method used in the proof of Lemma \ref{cdu_bound_lemma} is not applicable when $\gamma\ne 0$, and it is not easy to show a lower bound of the $c$-differential uniformity of $F$. But we confirm that $_c\Delta_F \ge 3$ when $7< p^n <500$  by using SageMath.

\subsection{The differential uniformity}

In this subsection we investigate the differential uniformity of $F$, that is, the case $c=1$. We first investigate the number of solutions of $D_aF(x)=b$ in $P_a$ when $\#P_a=6$.

\begin{lemma}\label{du_pa6_lemma} Let $a\in\Fpn\setminus \{0,\pm1,\pm\gamma,\pm(\gamma-1)\}$ and $b\in\Fpn$. Then, $D_aF(x)=b$ has four solutions in $P_a$ if and only if $\gamma=-1$ and at least one of the followings is satisfied:\\
(i) $p^n\equiv \pm1 \pmod{5}$, $a^4-3a^2+1=0$ and $b=a^{-1}$,\\
(ii) $p^n\equiv \pm1 \pmod{8}$, $a^2=2$ and $b=a$.\\
Otherwise, $D_aF(x)=b$ has at most three solutions.
\end{lemma}
\begin{proof} $D_aF(x)$ for each $x\in P_a$ is as follows:
\begin{equation*}
\begin{array}{ll}
D_aF(0)=F(a)-F(0)=a^{-1}, & D_aF(-a)=F(0)-F(-a)=a^{-1},\\
D_aF(1)=F(1+a)-F(1)=(a+1)^{-1}-\gamma^{-1}, & D_aF(1-a)=F(1)-F(1-a)=\gamma^{-1}-(1-a)^{-1},\\
D_aF(\gamma)=F(\gamma+a)-F(\gamma)=(a+\gamma)^{-1}-1, & D_aF(\gamma-a)=F(\gamma)-F(\gamma-a)=1+(a-\gamma)^{-1}.
\end{array}
\end{equation*}
We first assume that $b=D_aF(0)=D_aF(-a)=a^{-1}$. We investigate $D_aF(x)=a^{-1}$ has at least three solutions in $P_a$. First we investigate equivalent condition for each case that $D_aF(x)=a^{-1}$ has at least two solutions in $P_a$.
\begin{equation*}
\begin{array}{ll}
a^{-1}=D_aF(1)\ \Leftrightarrow\ a^2+a+\gamma=0, & 
a^{-1}=D_aF(1-a)\ \Leftrightarrow\ a^2-a+\gamma=0,\\
a^{-1}=D_aF(\gamma)\ \Leftrightarrow\ a^2+\gamma a+\gamma=0, & 
a^{-1}=D_aF(\gamma-a)\ \Leftrightarrow\ a^2-\gamma a+\gamma=0.
\end{array}
\end{equation*}
We can easily see that if $a^{-1}=D_aF(1)=D_aF(1-a)$ or $a^{-1}=D_aF(\gamma)=D_aF(\gamma-a)$ then we have $a=0$, a contradiction. Hence $D_aF(x)=a^{-1}$ has at most four solutions. If $a^{-1}=D_aF(1)=D_aF(\gamma-a)$  then we have $\gamma=-1$ and hence $a^2+a-1=0$, such $a$ exists if and only if $1^2-4\cdot(-1)=5$ is a square which is equivalent to $p^n\equiv 0,\pm 1\pmod{5}$. But if $p=5$ then we have $a=2=1-\gamma$ which contradicts to assumption, and hence we have $p\ne 5$ . Similarly, if $a^{-1}=D_aF(\gamma)=D_aF(1-a)$ we also have $\gamma=-1$ and hence $a^2-a-1=0$, such $a$ exists if and only if $(-1)^2-4\cdot(-1)=5$ is a square, and we have $p^n\equiv \pm 1 \pmod{5}$. Hence, $D_aF(x)=a^{-1}$ has four solutions if and only if $p^n\equiv \pm 1 \pmod{5}$, $\gamma=-1$ and $0=(a^2+a-1)(a^2-a-1)=a^4-3a^2+1$. Otherwise, $D_aF(x)=a^{-1}$ has at most three solutions. Conversely, if $p\equiv \pm 1\pmod{5}$ then there exists $a\in\Fpn$ with $a^2+a-1=0$ or $a^2-a-1=0$ and then by the above analysis we can see that $D_aF(x)=a^{-1}$ has four solutions in $P_a$ when $\gamma=-1$.

Next we investigate the case $b\ne a^{-1}$. In this case, $D_aF(x)=b$ has four solutions in $P_a$ if and only if $b=D_aF(1)= D_aF(\gamma)=D_aF(1-a)= D_aF(\gamma-a)$. First we investigate equivalent condition for each case that $D_aF(x)=b$ has at least two solutions in $\{1,\gamma,1-a,\gamma-a\}$.
\begin{equation*}
\begin{array}{ll}
D_aF(1)= D_aF(\gamma) \Leftrightarrow a^2+(\gamma+1) a+2\gamma=0, & 
D_aF(1)= D_aF(\gamma-a) \Leftrightarrow a(\gamma+1)(a-\gamma+1) a=0,\\
D_aF(1)= D_aF(1-a) \Leftrightarrow a^2+\gamma-1=0, &
D_aF(\gamma)= D_aF(1-a) \Leftrightarrow a(\gamma+1)(a+\gamma-1) a=0, \\ 
D_aF(\gamma)= D_aF(\gamma-a) \Leftrightarrow  a^2+\gamma(1-\gamma)=0, &
D_aF(1-a)= D_aF(\gamma-a) \Leftrightarrow a^2-(\gamma+1) a+2\gamma=0.
\end{array}
\end{equation*}
From $a^2+(\gamma+1) a+2\gamma=0$ and $a^2-(\gamma+1) a+2\gamma=0$ we have $\gamma=-1$. Then we have $a^2=2$, such $a$ exists if and only if $p^n\equiv \pm 1\pmod{8}$. Then we have $b=D_aF(1)=\frac{1}{a+1}+1=\frac{a+2}{a+1}=\frac{a^2+a-2}{a^2-1}=a$. Conversely, if $p^n\equiv 1\pmod{8}$ then there exists $a\in\Fpn$ with $a^2=2$ and then by the above analysis we can see that $D_aF(x)=a$ has four solutions in $P_a$, when $\gamma=-1$.
\end{proof}

Next we show the main theorem of this subsection.

\begin{theorem}\label{du_thm_1r} Let $F=Inv \circ (1,\gamma)$ with $\gamma\ne 0,1$. Then,\\
(i) $\Delta_F=7$ if and only if $\gamma=-1$, $p=5$ and $n$ is even.\\
(ii) $\Delta_F=6$ if and only if $\gamma=-1$, and $p^n\equiv 1 \pmod{8}$ or $p^n\equiv 1,4\pmod{15}$.\\
Otherwise, $\Delta_F\le 5$.
\end{theorem}

\begin{proof} 
Using Lemma \ref{panot6_lemma} and Lemma \ref{du_pa6_lemma} we investigate all the cases that $D_aF(x)=b$ has at least four solutions, which requires $\gamma=-1$. By Lemma \ref{panot2_lemma_1r} with $c=1$ and $\gamma=-1$, $D_aF(x)=b$ has two solutions in $\Fpn\setminus P_a$ if and only if $(ab)^2-4ab$ is a nonzero square, $(b+1)a+b=ab+a+b\ne 0$ and $(1-b)a+b=a+b-ab\ne0$.
\begin{itemize}
\item If $p=5$ and $(a,b)\in\{(2,3),(3,2)\}$ then $D_aF(x)=b$ has five solutions in $P_a$. Then $(ab)^2-4ab=1^2-4\cdot 1 =-3$ is a square if and only if $5^n\equiv 1\pmod{3}$ which is equivalent to $n$ is even. And, we have $ab+a+b=1\ne0$ and $a+b-ab=-1\ne 0$. Hence we have $\Delta_F(a,b)=7$ if $n$ is even, $\Delta_F(a,b)=5$ if $n$ is odd.
\item If $p=5$ and $(a,b)\in\{(1,-1),(-1,1)\}$ then $D_aF(x)=b$ has four solutions in $P_a$. Then we have $(ab)^2-4ab=1^2+4\cdot 1 =5=0$ and hence we have $\Delta_F(a,b)\le 5$ in this case.
\item If $p^n\equiv \pm 1 \pmod{5}$ and $a^4-3a^2+1=(a^2+a-1)(a^2-a-1)=0$ then $D_a(x)=a^{-1}$ has four solutions in $P_a$. Then we require $(ab)^2-4ab=1^2-4\cdot 1 =-3$ is a nonzero square if and only if $p^n\equiv 1\pmod {3}$. If $a^2+a-1=0$ then $ab+a+b=a^{-1}(a^2+a+1)=2a^{-1}\ne 0$ and $a+b-ab=a^{-1}(a^2-a+1)=2a^{-1}(1-a)\ne0$ since $a\ne 1$. Similarly, if $a^2-a-1=0$ then $ab+a+b\ne 0$ and $a+b-ab\ne0$. Hence we have $\Delta_F(a,a^{-1})=6$ where $a^4-3a^2+1=0$,  when $p^n\equiv \pm 1\pmod{5}$ and $p^n\equiv 1\pmod{3}$, or equivalently $p^n\equiv 1,4\pmod{15}$.
\item If $p^n\equiv \pm 1\pmod{8}$ and $a^2=2$ then $D_a(x)=a$ has four solutions in $P_a$. Then we have $(ab)^2-4ab=a^4-4a^2 =4-8=-4$ is a square if and only if $p^n \equiv 1\pmod{4}$. Since $a\ne 1$ from $(\pm 1)^2=1\ne 2$, we have $ab+a+b=a^2+2a=2(a+1)\ne 0$ and $a+b-ab=2a-2=2(a-1)\ne0$. Hence we have $\Delta_F(a,a)=6$ where $a^2=2$, when $p^n\equiv \pm 1\pmod{8}$ and $p^n\equiv 1\pmod{4}$, or equivalently $p^n\equiv 1\pmod{8}$.
\end{itemize}
Otherwise, $D_aF(x)=b$ at most three solutions in $P_a$ by Lemma \ref{panot6_lemma} and Lemma \ref{du_pa6_lemma}, and hence we have $\Delta_F\le 5$ by Lemma \ref{panot2_lemma_1r}.
\end{proof}

\subsection{The $c$-differential uniformity with $c\ne 0,1$}

In the above subsection we investigate the differential uniformity of $Inv\circ (1,\gamma)$. In this subsection we investigate the $c$-differential uniformity of $Inv\circ (1,\gamma)$ where $c\ne 0,1$. First we investigate the number of solutions of $_cD_aF(x)=b$ in $P_a$ when $\#P_a=6$.

\begin{lemma}\label{pa6_lemma} Let $F=Inv\circ (1,\gamma)$ with $\gamma\ne 0,1$ and $c\ne 0,1$. Let $a\in\Fpn\setminus \{0,\pm1,\pm\gamma,\pm(\gamma-1)\}$ and $b\in\Fpn$. Then, $_cD_aF(x)=b$ has four solutions in $P_a$ if and only if $\gamma c^2+2(\gamma^2+\gamma+1)c+\gamma=0$, $c\ne -1$, $a=\frac{\gamma^2-c}{\gamma+c}$ and $b=\frac{1-c}{1+\gamma}$. Otherwise, $_cD_aF(x)=b$ has at most three solutions in $P_a$.
\end{lemma}
\begin{proof} $_cD_aF(x)$ for each $x\in P_a$ is as follows :
\begin{equation*}
\begin{array}{ll}
_cD_aF(0)=F(a)-cF(0)=a^{-1}, & _cD_aF(-a)=F(0)-cF(-a)=ca^{-1},\\
_cD_aF(1)=F(1+a)-cF(1)=(a+1)^{-1}-c\gamma^{-1}, & _cD_aF(1-a)=F(1)-cF(1-a)=\gamma^{-1}-c(1-a)^{-1},\\
_cD_aF(\gamma)=F(\gamma+a)-cF(\gamma)=(a+\gamma)^{-1}-c, & _cD_aF(\gamma-a)=F(\gamma)-cF(\gamma-a)=1+c(a-\gamma)^{-1}.
\end{array}
\end{equation*}
We investigate some possible cases to be used in this proof such that $_cD_aF(x)=b$ has at least two solutions in $P_a$ and we obtain an equation on $a$ in each case.
\begin{equation*}
\begin{array}{ll}
 _cD_aF(0)={}_cD_aF(1) \Leftrightarrow ca^2+ca+\gamma=0, &
 _cD_aF(0)={}_cD_aF(1-a) \Leftrightarrow a^2+(c\gamma-\gamma-1)a+\gamma=0,\\
 _cD_aF(0)={}_cD_aF(\gamma) \Leftrightarrow ca^2+c\gamma a+\gamma=0,&
 _cD_aF(0)={}_cD_aF(\gamma-a) \Leftrightarrow a^2+(c-\gamma-1)a+\gamma=0,\\
_cD_aF(-a)={}_cD_aF(1-a) \Leftrightarrow a^2 - a+c\gamma =0,&
 _cD_aF(1)={}_cD_aF(-a) \Leftrightarrow ca^2+(c\gamma+c-\gamma) a+c\gamma=0,\\
 _cD_aF(-a)={}_cD_aF(\gamma-a) \Leftrightarrow a^2 - \gamma a+c\gamma =0,&
 _cD_aF(\gamma)={}_cD_aF(-a) \Leftrightarrow c a^2+(c\gamma+c-1) a+c\gamma=0.
\end{array}
\end{equation*}
We investigate possible cases such that $_cD_aF(x)=b$ has at least four solutions in $P_a$. Since $c\ne 1$ we have $_cD_aF(0)\ne {}_cD_aF(-a)$, and hence $_cD_aF(x)=b$ has at most one solution in $\{0,-a\}$. We first investigate the cases that $_cD_aF(x)=b$ has one solution in $\{0,-a\}$.
\begin{itemize}
\item If $_cD_aF(0)={}_cD_aF(1)={}_cD_aF(\gamma)$ then $_cD_aF(0)={}_cD_aF(1)$ and $_cD_aF(0)={}_cD_aF(\gamma)$, and hence we have $ca^2+ca+\gamma=0$ and $ca^2+c\gamma a+\gamma=0$. Then we have $\gamma=1$, a contradiction, so this case cannot happen.
\item If $_cD_aF(0)={}_cD_aF(1-a)={}_cD_aF(\gamma-a)$ then $_cD_aF(0)={}_cD_aF(1-a)$ and $_cD_aF(0)={}_cD_aF(\gamma-a)$, and hence we have $a^2+(c\gamma-\gamma-1)a+\gamma=0$ and $a^2+(c-\gamma-1)a+\gamma=0$. Then we have $\gamma=1$, a contradiction, so this case cannot happen.
\item $_cD_aF(1)={}_cD_aF(\gamma)={}_cD_aF(-a)$ then $_cD_aF(1)={}_cD_aF(-a)$ and $_cD_aF(\gamma)={}_cD_aF(-a)$, and hence $ca^2+(c\gamma+c-\gamma) a+c\gamma=0$ and $c a^2+(c\gamma+c-1) a+c\gamma=0$. Then we have $\gamma=1$, a contradiction, so this case cannot happen.
\item If $_cD_aF(-a)={}_cD_aF(1-a)={}_cD_aF(\gamma-a)$ then $_cD_aF(-a)={}_cD_aF(1-a)$ and $_cD_aF(-a)={}_cD_aF(\gamma-a)$, and hence we have $a^2 - a+c\gamma =0$ and $a^2 - \gamma a+c\gamma =0$. Then we have $\gamma=1$, a contradiction, so this case cannot happen.
\end{itemize}
In this case, $_cD_aF(x)=b$ has at least four solutions in $P_a$ requires it has at least three solutions in $\{1,\gamma, 1-a, \gamma-a\}$. Then at least one case that is considered above should be satisfied, which is a contradiction. Hence the only possible case (that $_cD_aF(x)=b$ has four solutions) is $b={}_cD_aF(1)={}_cD_aF(\gamma)={}_cD_aF(1-a)={}_cD_aF(\gamma-a)$. We first investigate  all possible cases that $_cD_aF(x)=b$ has at least two solutions in $\{1,\gamma, 1-a, \gamma-a\}$.
\begin{itemize}
\item $ _cD_aF(1)={}_cD_aF(\gamma)\ \Leftrightarrow\ ca^2+c(\gamma+1) a+\gamma(c+1)=0$.
\item $ _cD_aF(1)={}_cD_aF(1-a)\ \Leftrightarrow\ (c+1)a^2+\gamma(c-1) a+(\gamma-1)(c+1)=0$.
\item $_cD_aF(1)={}_cD_aF(\gamma-a)\ \Leftrightarrow\ (c+\gamma)a^2+(c-\gamma^2) a=0$ and hence $a=\frac{\gamma^2-c}{\gamma+c}$ when $c\ne -\gamma$.
\item $_cD_aF(\gamma)={}_cD_aF(1-a)\ \Leftrightarrow\ (c\gamma+1) a^2+(c\gamma^2-1) a=0$ and hence $a=\frac{1-c\gamma^2}{1+c\gamma}$ when $c\ne -\gamma^{-1}$.
\item $_cD_aF(\gamma)={}_cD_aF(\gamma-a)\ \Leftrightarrow\ (c+1) a^2+(c-1) a+\gamma(1-\gamma)(c+1)=0$
\end{itemize}
If $c=-1$, then $_cD_aF(1)={}_cD_aF(1-a)$ and $_cD_aF(\gamma)={}_cD_aF(\gamma-a)$ imply $a=0$, a contradiction. Hence we require $c\ne -1$. Next we investigate all possible cases that $_cD_aF(x)=b$ has at least three solutions in $\{1,\gamma, 1-a, \gamma-a\}$.
\begin{itemize}
\item If $_cD_aF(1)={}_cD_aF(\gamma)={}_cD_aF(1-a)$ then we have $ca^2+c(\gamma+1) a+\gamma(c+1)=0$, $(c+1)a^2+\gamma(c-1) a+(\gamma-1)(c+1)=0$ and $(c\gamma+1) a^2+(c\gamma^2-1) a=0$. We substitute $ca^2=-c(\gamma+1) a-\gamma(c+1)$ which is obtained from the first equation to the second equation to have $a^2-(c+\gamma)a-(c+1)=0$. We substitute $a=\frac{1-c\gamma^2}{c\gamma+1}$ to $a^2-(c+\gamma)a-(c+1)=0$ to have $(\gamma^3-\gamma^2)c^3+(2\gamma^4-3\gamma)c^2+(\gamma^3-3\gamma^2-2\gamma-2)c-\gamma=0$ when $c\ne -\gamma^{-1}$. But if $c=-\gamma^{-1}$, then we have $0=(\gamma^3-\gamma^2)c^3+(2\gamma^4-3\gamma)c^2+(\gamma^3-3\gamma^2-2\gamma-2)c-\gamma=(\gamma+1)^2$ and then $\gamma=-1$ and $c=1$, a contradiction.
\item If $_cD_aF(1)={}_cD_aF(\gamma)={}_cD_aF(\gamma-a)$ then we have $ca^2+c(\gamma+1) a+\gamma(c+1)=0$, $(c+\gamma)a^2+(c-\gamma^2) a=0$ and $(c+1) a^2+(c-1) a+\gamma(1-\gamma)(c+1)=0$. We substitute $ca^2=-c(\gamma+1) a-\gamma(c+1)$ which is obtained from the first equation to the third equation to have $a^2-(c\gamma+1)a-\gamma^2(c+1)=0$. We substitute $a=\frac{\gamma^2-c}{\gamma+c}$ from the second equation to $a^2-(c\gamma+1)a-\gamma^2(c+1)=0$ to have $(\gamma^2-\gamma)c^3+(3\gamma^3-2)c^2+(2\gamma^4+2\gamma^3+3\gamma^2-\gamma)c+\gamma^3=0$ when $c\ne -\gamma$. But if $c=-\gamma$ then $0=(\gamma^2-\gamma)c^3+(3\gamma^3-2)c^2+(2\gamma^4+2\gamma^3+3\gamma^2-\gamma)c+\gamma^3=-\gamma^2(\gamma+1)^2$ which implies $\gamma=-1$ and $c=1$, a contradiction.
\item If $_cD_aF(1)={}_cD_aF(1-a)={}_cD_aF(\gamma-a)$ then we have $_{c^{-1}}D_{-a}F(1-(-a))={}_{c^{-1}}D_{-a}F(1-a-(-a))={}_{c^{-1}}D_{-a}F(\gamma-a-(-a))$ by Lemma \ref{cdu_symm_lemma}, or equivalently $_{c^{-1}}D_{a'}F(1)={}_{c^{-1}}D_{a'}F(\gamma)={}_{c^{-1}}D_{a'}F(1-a')$ where $a'=-a$. By the above analysis, we have $(\gamma^3-\gamma^2)(c^{-1})^3+(2\gamma^4-3\gamma)(c^{-1})^2+(\gamma^3-3\gamma^2-2\gamma-2)c^{-1}-\gamma=0$ or equivalently $\gamma c^3-(\gamma^3-3\gamma^2-2\gamma-2)c^2-(2\gamma^4-3\gamma)c-\gamma^3+\gamma^2=0$.
\item If $_cD_aF(\gamma)={}_cD_aF(1-a)={}_cD_aF(\gamma-a)$ then we have $_{c^{-1}}D_{a'}F(1)={}_{c^{-1}}D_{a'}F(\gamma)={}{c^{-1}}D_{a'}F(\gamma-a')$ where $a'=-a$, similarly with the above case. Then we have $(\gamma^2-\gamma)(c^{-1})^3+(3\gamma^3-2)(c^{-1})^2+(2\gamma^4+2\gamma^3+3\gamma^2-\gamma)c^{-1}+\gamma^3=0$ or equivalently $\gamma^3 c^3+(2\gamma^4+2\gamma^3+3\gamma^2-\gamma)c^2+(3\gamma^3-2)c+\gamma^2-\gamma=0$.
\end{itemize}
So we have the following four cubic equations on $c$.
\begin{align}
(\gamma^3-\gamma^2)c^3+(2\gamma^4-3\gamma)c^2+(\gamma^3-3\gamma^2-2\gamma-2)c-\gamma&=0, \label{inv1r_eq1}\\
(\gamma^2-\gamma)c^3+(3\gamma^3-2)c^2+(2\gamma^4+2\gamma^3+3\gamma^2-\gamma)c+\gamma^3&=0, \label{inv1r_eq2}\\
\gamma c^3-(\gamma^3-3\gamma^2-2\gamma-2)c^2-(2\gamma^4-3\gamma)c-\gamma^3+\gamma^2&=0, \label{inv1r_eq3}\\
\gamma^3 c^3+(2\gamma^4+2\gamma^3+3\gamma^2-\gamma)c^2+(3\gamma^3-2)c+\gamma^2-\gamma&=0.\label{inv1r_eq4}
\end{align}
Then $\eqref{inv1r_eq2}-(\gamma-1)\times\eqref{inv1r_eq3}$ is the following quadratic equation.
\begin{align}
 0&=(\gamma^4-\gamma^3+\gamma^2)c^2+2(\gamma^5+\gamma^3+\gamma)c+\gamma^4-\gamma^3+\gamma^2\notag \\
&=\gamma(\gamma^2-\gamma+1)\left( \gamma c^2+2(\gamma^2+\gamma+1)c+\gamma \right) \label{inv1r_eq23}
\end{align}
We also confirm that $-\eqref{inv1r_eq1}+\gamma \times\eqref{inv1r_eq2}$ and $\left(-\eqref{inv1r_eq3}\times \gamma^2+\eqref{inv1r_eq4}\right)/(\gamma-1)$ are same with $\eqref{inv1r_eq23}\times (\gamma+1)/\gamma$. If $\gamma^2-\gamma+1=0$ then equations \eqref{inv1r_eq1}, \eqref{inv1r_eq2}, \eqref{inv1r_eq3} and \eqref{inv1r_eq4} imply $0=c^3+5c^2+5c+1=(c+1)(c^2+4c+1)$. Since $c\ne-1$, we have $c^2+4c+1=0$ and then $\gamma c^2+2(\gamma^2+\gamma+1)c+\gamma=2c(\gamma^2-\gamma +1)=0$ when $\gamma^2-\gamma+1=0$. Hence we only require that $\gamma c^2+2(\gamma^2+\gamma+1)c+\gamma=0$ whether $\gamma^2-\gamma+1=0$ or not. In this case we set $a=\frac{\gamma^2-c}{\gamma+c}$ and $b={}_cD_aF(1)=(a+1)^{-1}-c\gamma^{-1}=\frac{1-c}{1+\gamma}$.

Conversely, we assume that $\gamma c^2+2(\gamma^2+\gamma+1)c+\gamma=0$ with $c=-1$. If $c=-\gamma$ then we have $0=-\gamma^3-2\gamma^2-\gamma=-\gamma(\gamma+1)^2$ so $\gamma=-1$ and hence $c=-\gamma=1$ which contradicts to the assumption, so we have $c\ne-\gamma$. Similarly, we have $c\ne -\gamma^{-1}$. Since $(\gamma^2-c)(1+c\gamma)-(\gamma+c)(1-c\gamma^2)=(\gamma-1)\left(\gamma c^2+2(\gamma^2+\gamma+1)c+\gamma\right)=0$ we set $a=\frac{\gamma^2-c}{\gamma+c}=\frac{1-c\gamma^2}{1+c\gamma}$. By the above analysis, we have $_cD_aF(1)={}_cD_aF(\gamma-a)$ and $_cD_aF(\gamma)={}_cD_aF(1-a)$. Since $a=\frac{\gamma^2-c}{\gamma+c}$, we have $(c+1)a^2+\gamma(c-1) a+(\gamma-1)(c+1)=\frac{\gamma(\gamma-1)\left(\gamma c^2+2(\gamma^2+\gamma+1)c+\gamma\right)}{(c+\gamma)^2}=0$ and hence $_cD_aF(1)={}_cD_aF(1-a)$. Hence $_cD_aF(x)=b$ has four solutions in $P_a$, which completes the proof.
\end{proof}

Next we show the main theorem of this subsection.

\begin{theorem}\label{cdu_thm_1r} Let $F=Inv\circ (1,\gamma)$ with $\gamma\ne 0,1$ and $c\ne 0,1$. Then $_c\Delta_F =6$ if and only if $\gamma c^2+2(\gamma^2+\gamma+1)c+\gamma=0$, $c\ne-1$, and $-c$ is a square. Otherwise, $_c\Delta_F \le 5$.
\end{theorem}

\begin{proof} By Lemma \ref{panot6_lemma} and Lemma \ref{pa6_lemma}, $_cD_aF(x)=b$ has four solutions in $P_a$ if and only if 
\begin{equation}\label{cdu_1r_cond}
\gamma c^2+2(\gamma^2+\gamma+1)c+\gamma=0,\ c\ne-1,\ a=\frac{\gamma^2-c}{\gamma+c},\ b=\frac{1-c}{1+\gamma},
\end{equation}
otherwise  $_cD_aF(x)=b$ has at most three solutions in $P_a$. Hence it remains to investigate required conditions that $_cD_aF(x)=b$ has two solutions in $\Fpn\setminus P_a$ using Lemma \ref{panot2_lemma_1r}, under the assumption as in \eqref{cdu_1r_cond}. We require that
\begin{align*}
(ab+c-1)^2-4abc&=\frac{\gamma^2c^4+4(\gamma^4-\gamma-1)c^3+2(2\gamma^2-\gamma^2+2)c^2-4(\gamma^4+\gamma^3-\gamma)c+\gamma^2}{(\gamma+1)^2(c+\gamma)^2}\\
&=\frac{\left(\gamma c^2+2(\gamma^2+\gamma+1)c+\gamma\right)^2-4c(\gamma+1)^2(c+\gamma)^2}{(\gamma+1)^2(c+\gamma)^2}=-4c
\end{align*}
is a square. We can see that $c\ne -\gamma$ in the proof of Lemma \ref{pa6_lemma}. If $\gamma=-1$ then we have $0=\gamma c^2+2(\gamma^2+\gamma+1)c+\gamma=-(c-1)^2$, and then $c=1$ which is a contradiction, so we have $\gamma\ne -1$. Hence we get $(b+c)a+b+c-1=\frac{c(\gamma^2-1)}{c+\gamma}\ne 0$, $(1-b)a+b+c-1=\frac{(\gamma-1)(c+\gamma)}{\gamma+1}\ne0$, $(b\gamma+c)a+\gamma(b\gamma+c-1)\frac{(\gamma-1)(c+\gamma)}{\gamma+1}\ne0$ and 
$(1-b\gamma)a+\gamma(b\gamma+c-1)=\frac{c(\gamma^2-1)}{c+\gamma}\ne 0$. Therefore, under the assumption \eqref{cdu_1r_cond} $_cD_aF(x)=b$ has two solutions in $\Fpn\setminus P_a$ if and only if $-c$ is a square, by Lemma \ref{panot2_lemma_1r}.
\end{proof}

\section{Concluding Remark}

In this paper, we studied the $c$-differential
uniformity of the swapped inverse function $Inv \circ (1,\gamma)$.
In section \ref{sec_inv01}, we showed that the $c$-differential
uniformity of $Inv \circ (0,1)$ is at most $5$. In section
\ref{sec_inv1r}, we showed that the $c$-differential uniformity of
$Inv \circ (1,\gamma)$ is at most $6$ when $\gamma\ne 1$. For each
case (i.e., $c=1$ or $c\neq 1$), we gave a  detailed
classification of the ($c$-)differential uniformity of $Inv \circ
(1,\gamma)$, which has been unknown as far as we know.
Unfortunately we could not  find (A)PcN permutations except when
$p^n$ is very small. In future studies, we will continue to find
permutations $Inv \circ h $ with cyclic permutation $h$ having low
$c$-differential uniformities.

\bigskip
\noindent \textbf{Acknowledgements} : This work was supported by the National Research Foundation of Korea (NRF) grant
funded by the Korea government (MSIT) (No. 2021R1C1C2003888). Soonhak Kwon was supported by the National Research Foundation of Korea (NRF) grant funded by the Korea government (MSIT) (No. 2016R1A5A1008055, No. 2019R1F1A1058920 and 2021R1F1A1050721).

\begin{appendix}

\section{Proof of Lemma \ref{panot6_lemma}}\label{proof_panot6_lemma}

By Lemma \ref{cdu_symm_lemma}, we have $_c\Delta_F(-a,b)={}_{c^{-1}}\Delta_F(a,-bc^{-1})$. Hence it is enough to investigate the case $a\in\{1,\gamma,\gamma-1\}$.

\bigskip
\noindent If $a=1$ with $\gamma=-1$ then $P_a=\{0,\pm1, -2\}$ and we have $_cD_1F(0)=-1$, $_cD_1F(1)=2^{-1}+c$, $_cD_1F(-1)=-c$ and $_cD_1F(-2)=1+2^{-1}c$. 
\begin{equation*}
\begin{array}{ll}
_cD_1F(0)={}_cD_1F(1)\ \Leftrightarrow\ c=-\frac{3}{2}, &  _cD_1F(0)={}_cD_1F(-1)\ \Leftrightarrow\ c=1,\\
_cD_1F(0)={}_cD_1F(-2)\ \Leftrightarrow\ c=-4, & _cD_1F(1)={}_cD_1F(-1)\ \Leftrightarrow\ c=-\frac{1}{4},\\
_cD_1F(1)={}_cD_1F(-2)\ \Leftrightarrow\ c=1, &  _cD_1F(-1)={}_cD_1F(-2)\ \Leftrightarrow\ c=-\frac{2}{3}.
\end{array}
\end{equation*}
We investigate all possible cases such that $_cD_1F(x)=b$ has at least three solutions in $P_a$.
\begin{equation*}
\begin{array}{l}
_cD_1F(0)={}_cD_1F(1)={}_cD_1F(-1)\ \Leftrightarrow\ c=-\frac{3}{2}=1=-\frac{1}{4}\ \Rightarrow 2=-3,\ 6=1\text{ and }1=-4 \Rightarrow\ p=5.\\
_cD_1F(0)={}_cD_1F(1)={}_cD_1F(-2)\ \Leftrightarrow\ c=-\frac{3}{2}=-4=1\ \Rightarrow -8=-3,\ 2=-3\text{ and }1=-4 \Rightarrow\ p=5.\\
_cD_1F(0)={}_cD_1F(-1)={}_cD_1F(-2)\ \Leftrightarrow\ c=1=-4=-\frac{2}{3} \ \Rightarrow 1=-4,\ 2=-3\text{ and }-12=-2 \Rightarrow\ p=5.\\
_cD_1F(1)={}_cD_1F(-1)={}_cD_1F(-2)\ \Leftrightarrow\ c=-\frac{1}{4}=1=-4=-\frac{2}{3} \ \Rightarrow 1=-4=16 \Rightarrow\ p=5.
\end{array}
\end{equation*}
So we have $_cD_1F(0)={}_cD_1F(1)={}_cD_1F(-1)={}_cD_1F(-2)=-1$ when $c=1$ and $p=5$. 

\bigskip
\noindent  If $a=1$ with $\gamma=2$ then $P_a=\{0, \pm1, 2\}$. If $p=3$ then we have $P_a=\F_3$ and hence $_cD_aF(x)=b$ has at most three solutions in $P_a$. If $p\ne 3$, then we have $_cD_1F(0)=2^{-1}$, $_cD_1F(1)=1-2^{-1}c$, $_cD_1F(-1)=c$ and $_cD_1F(2)=3^{-1}-c$. We investigate equivalent condition for each case that $D_aF(x)=b$ has at least two solutions in $P_a$.
\begin{equation*}
\begin{array}{ll}
_cD_1F(0)={}_cD_1F(1)\ \Leftrightarrow\ c=1, &  _cD_1F(0)={}_cD_1F(-1)\ \Leftrightarrow\ c=2^{-1},\\
_cD_1F(0)={}_cD_1F(2)\ \Leftrightarrow\ c=-\frac{1}{6}, & _cD_1F(1)={}_cD_1F(-1)\ \Leftrightarrow\ c=\frac{2}{3},\\
_cD_1F(1)={}_cD_1F(2)\ \Leftrightarrow\ c=-\frac{4}{3}, &  _cD_1F(-1)={}_cD_1F(2)\ \Leftrightarrow\ c=\frac{1}{6}.
\end{array}
\end{equation*}
We investigate all possible cases such that $_cD_1F(x)=b$ has at least three solutions in $P_a$.
\begin{equation*}
\begin{array}{l}
_cD_1F(0)={}_cD_1F(1)={}_cD_1F(-1)\ \Leftrightarrow\ c=1=2^{-1}\ \Rightarrow 2=1\text{, a contradiction.}\\
_cD_1F(0)={}_cD_1F(1)={}_cD_1F(2)\ \Leftrightarrow\ c=1=-\frac{1}{6}=-\frac{4}{3}\ \Rightarrow 1=-6,\ 3=-4\text{ and }3=24 \Rightarrow\ p=7.\\
_cD_1F(0)={}_cD_1F(-1)={}_cD_1F(2)\ \Leftrightarrow\ c=\frac{1}{6}=-\frac{1}{6}\text{, a contradiction.}\\
_cD_1F(1)={}_cD_1F(-1)={}_cD_1F(2)\ \Leftrightarrow\ c=\frac{2}{3}=-\frac{4}{3}\ \Rightarrow 2=0\text{, a contradiction.}
\end{array}
\end{equation*}
So we can see that $_cD_1F(x)=b$ has at most three solutions.

\bigskip
\noindent If $a=1$ with $\gamma\ne -1,2$ then $P_a=\{0,1,\gamma, -1,\gamma-1\}$ and we have $_cD_{1}F(0)=\gamma^{-1}$, $_cD_{1}F(1)=2^{-1}-c\gamma^{-1}$, $_cD_{1}F(\gamma)=(\gamma+1)^{-1}-c$, $_cD_{1}F(-1)=c$ and $_cD_{1}F(\gamma-1)=1-c(\gamma-1)^{-1}$. We investigate equivalent condition for each case that $D_aF(x)=b$ has at least two solutions in $P_a$.
\begin{equation*}
\begin{array}{ll}
_cD_1F(0)={}_cD_1F(1)\ \Leftrightarrow\ c=\frac{\gamma-2}{2}, &  
_cD_1F(0)={}_cD_1F(\gamma)\ \Leftrightarrow\ c=-\frac{1}{\gamma(\gamma+1)},\\
_cD_1F(0)={}_cD_1F(-1)\ \Leftrightarrow\ c=\gamma^{-1}, & 
_cD_1F(0)={}_cD_1F(\gamma-1)\ \Leftrightarrow\ c=\frac{(\gamma-1)^2}{\gamma},\\
_cD_1F(1)={}_cD_1F(\gamma)\ \Leftrightarrow\ c=-\frac{\gamma}{2(\gamma+1)}, &  
_cD_1F(1)={}_cD_1F(-1)\ \Leftrightarrow\ c=\frac{2\gamma}{\gamma+1},  \\
_cD_1F(1)={}_cD_1F(\gamma-1)\ \Leftrightarrow\ c=\frac{\gamma(\gamma-1)}{2}, &  
_cD_1F(\gamma)={}_cD_1F(-1)\ \Leftrightarrow\ c=\frac{1}{2(\gamma+1)},\\
_cD_1F(\gamma)={}_cD_1F(\gamma-1)\ \Leftrightarrow\ c=-\frac{\gamma(\gamma-1)}{(\gamma+1)(\gamma-2)},  &  
_cD_1F(-1)={}_cD_1F(\gamma-1)\ \Leftrightarrow\ c=\frac{\gamma-1}{\gamma}. 
\end{array}
\end{equation*}
Next, we can see that $_cD_1F(x)=b$ cannot have four solutions in $P_a$ by showing all possible cases cannot happen.
\begin{itemize}
\item $_cD_1 F(0)={}_cD_1 F(1)={}_cD_1 F(-1)\ \Leftrightarrow\ c=\frac{\gamma-2}{2}=\gamma^{-1}=\frac{2\gamma}{\gamma+1}$. Then we have $\gamma^2-2\gamma-2=2\gamma^2-\gamma-1=\gamma^2-5\gamma-2=0$ and then $\gamma=0$, a contradiction. Hence $_cD_1 F(0)={}_cD_1 F(1)={}_cD_1 F(\gamma)={}_cD_1 F(-1)$ and $_cD_1 F(0)={}_cD_1 F(1)={}_cD_1 F(-1)={}_cD_1 F(\gamma-1)$ cannot happen. 
\item $_cD_1 F(1)={}_cD_1 F(\gamma)={}_cD_1 F(-1)\ \Leftrightarrow\ c=-\frac{\gamma}{2(\gamma+1)}=\frac{2\gamma}{\gamma+1}=\frac{1}{2(\gamma+1)}$. This case cannot happen, since  $-\frac{\gamma}{2(\gamma+1)}=\frac{1}{2(\gamma+1)}$ implies $\frac{1}{2}=0$, a contradiction. So $_cD_1 F(0)={}_cD_1 F(\gamma)={}_cD_1 F(-1)={}_cD_1 F(\gamma-1)$ is also impossible.
\item $_cD_1 F(1)={}_cD_1 F(\gamma)={}_cD_1 F(\gamma-1)\ \Leftrightarrow\ c=-\frac{\gamma}{2(\gamma+1)}=\frac{\gamma(\gamma-1)}{2}=-\frac{\gamma(\gamma-1)}{(\gamma+1)(\gamma-2)}$. This case cannot happen, since  $-\frac{\gamma}{2(\gamma+1)}=\frac{\gamma(\gamma-1)}{2}$ implies $\gamma=0$, a contradiction. Hence $_cD_1 F(0)={}_cD_1 F(1)={}_cD_1 F(\gamma)={}_cD_1 F(\gamma-1)$ and $_cD_1 F(1)={}_cD_1 F(\gamma)={}_cD_1 F(-1)={}_cD_1 F(\gamma-1)$ cannot happen, neither.
\end{itemize}

\noindent If $a=\gamma$ with $\gamma= -1$, then we have $_cD_{-1}F(1)={}_cD_{-1}F(2)={}_cD_{-1}F(0)={}_cD_{-1}F(-1)=1$ when $c=1$ and $p=5$ by applying Lemma \ref{cdu_symm_lemma} for the already investigated case $a=1$ and $\gamma=-1$. If $a=\gamma$ with $\gamma= 2^{-1}$ then $_cD_{2^{-1}}F_1(x)=b$ and $_cD_{1}F_2(x)=b$ have the same number of solutions in $P_a$ where $F_1=Inv \circ (1,2^{-1})$ and $F_2 = Inv \circ (1,2)$, by applying $\gamma=2$ in Lemma \ref{cdu_symm_lemma_1r}. Hence $_cD_{2^{-1}}F_1(x)=b$ has at most three solutions in $P_a$. If $a=\gamma$ with $\gamma\ne -1, 2^{-1}$ then $P_a=\{0,1,\gamma,-\gamma,1-\gamma\}$, and we have $_cD_{\gamma}F(0)=1$, $_cD_{\gamma}F(1)=(\gamma+1)^{-1}-c\gamma^{-1}$, $_cD_{\gamma}F(\gamma)=(2\gamma)^{-1}-c$, $_cD_{\gamma}F(-\gamma)=c\gamma^{-1}$ and $_cD_{\gamma}F(1-\gamma)=\gamma^{-1}+c(\gamma-1)^{-1}$. We investigate equivalent condition for each case that $D_aF(x)=b$ has at least two solutions in $P_a$.
\begin{equation*}
\begin{array}{ll}
_cD_{\gamma}F(0)={}_cD_{\gamma}F(1)\ \Leftrightarrow\ c=-\frac{\gamma^2}{\gamma+1}, &  
_cD_{\gamma}F(0)={}_cD_{\gamma}F(\gamma)\ \Leftrightarrow\ c=-\frac{2\gamma-1}{2\gamma},\\
_cD_{\gamma}F(0)={}_cD_{\gamma}F(-\gamma)\ \Leftrightarrow\ c=\gamma, & 
_cD_{\gamma}F(0)={}_cD_{\gamma}F(1-\gamma)\ \Leftrightarrow\ c=\frac{(\gamma-1)^2}{\gamma},\\
_cD_{\gamma}F(1)={}_cD_{\gamma}F(\gamma)\ \Leftrightarrow\ c=-\frac{1}{2(\gamma+1)}, &  
_cD_{\gamma}F(1)={}_cD_{\gamma}F(-\gamma)\ \Leftrightarrow\ c=\frac{\gamma}{2(\gamma+1)}, \\
_cD_{\gamma}F(1)={}_cD_{\gamma}F(1-\gamma)\ \Leftrightarrow\ c=-\frac{\gamma-1}{(\gamma+1)(2\gamma-1)}, &  
_cD_{\gamma}F(\gamma)={}_cD_{\gamma}F(-\gamma)\ \Leftrightarrow\ c=\frac{1}{2(\gamma+1)},\\
_cD_{\gamma}F(\gamma)={}_cD_{\gamma}F(1-\gamma)\ \Leftrightarrow\ c=-\frac{\gamma-1}{2\gamma^2},  &  
_cD_{\gamma}F(-\gamma)={}_cD_{\gamma}F(1-\gamma)\ \Leftrightarrow\ c=-\gamma+1. 
\end{array}
\end{equation*}
Next, we can see that $_cD_1F(x)=b$ cannot have four solutions in $P_a$ by showing all possible cases cannot happen.
\begin{itemize}
\item $_cD_{\gamma}F(0)={}_cD_{\gamma}F(-\gamma)={}_cD_{\gamma}F(1-\gamma)\ \Leftrightarrow\ c=\gamma =\frac{(\gamma-1)^2}{\gamma}=-\gamma+1$. Then we have $c=\gamma=2^{-1}$, a contradiction. Hence $_cD_{\gamma}F(0)={}_cD_{\gamma}F(1)={}_cD_{\gamma}F(-\gamma)={}_cD_{\gamma}F(1-\gamma)$ and $_cD_{\gamma}F(0)={}_cD_{\gamma}F(\gamma)={}_cD_{\gamma}F(-\gamma)={}_cD_{\gamma}F(1-\gamma)$ are also impossible.
\item $_cD_{\gamma}F(1)={}_cD_{\gamma}F(\gamma)={}_cD_{\gamma}F(-\gamma)\ \Leftrightarrow\ c=-\frac{1}{2(\gamma+1)}=\frac{\gamma}{2(\gamma+1)}=\frac{1}{2(\gamma+1)}$. This case cannot happen, because we have $\frac{1}{2}=0$ from $-\frac{1}{2(\gamma+1)}=\frac{\gamma}{2(\gamma+1)}$ a contradiction. Hence $_cD_{\gamma}F(0)={}_cD_{\gamma}F(1)={}_cD_{\gamma}F(\gamma)={}_cD_{\gamma}F(-\gamma)$ and  $_cD_{\gamma}F(1)={}_cD_{\gamma}F(\gamma)={}_cD_{\gamma}F(-\gamma)={}_cD_{\gamma}F(1-\gamma)$ are also impossible.
\item $_cD_{\gamma}F(1)={}_cD_{\gamma}F(\gamma)={}_cD_{\gamma}F(1-\gamma)\ \Leftrightarrow\ c=-\frac{1}{2(\gamma+1)}=-\frac{\gamma-1}{(\gamma+1)(2\gamma-1)}=-\frac{\gamma-1}{2\gamma^2}$. This case cannot happen, because we have $-1=0$ from $-\frac{1}{2(\gamma+1)}=-\frac{\gamma-1}{2\gamma^2}$, a contradiction. So, $_cD_{\gamma}F(0)={}_cD_{\gamma}F(1)={}_cD_{\gamma}F(\gamma)={}_cD_{\gamma}F(1-\gamma)$ is also impossible.
\end{itemize}

\noindent If $a=\gamma-1$ with $\gamma=2$ then we have $a=1$ which is already investigated. If $a=\gamma-1$ with $\gamma=2^{-1}$ then we have $a=-2^{-1}$. By Lemma \ref{cdu_symm_lemma}, $_cD_{-2^{-1}}F(x)=b$ and $_{c^{-1}}D_{2^{-1}}F(x)=-b$ have the same number of solutions in $P_a$. By Lemma \ref{cdu_symm_lemma_1r}, $_{c^{-1}}D_{2^{-1}}F(x)=-b$ and  $_{c^{-1}}D_1F_2(x)=-2^{-1}b$ have the same number of solutions in $P_a$ where $F_2 = Inv \circ (1,2)$. Hence $_cD_{-2^{-1}}F(x)=b$ has at most three solutions in $P_a$. If $a=\gamma-1$ with $\gamma\ne 2,2^{-1}$ then $P_a=\{0,1,\gamma,1-\gamma,2-\gamma\}$ and we have $_cD_{\gamma-1}F(0)=(\gamma-1)^{-1}$, $_cD_{\gamma-1}F(1)=1-c\gamma^{-1}$, $_cD_{\gamma-1}F(\gamma)=(2\gamma-1)^{-1}-c$, $_cD_{\gamma-1}F(1-\gamma)=c(\gamma-1)^{-1}$ and $_cD_{\gamma-1}F(2-\gamma)=\gamma^{-1}+c(\gamma-2)^{-1}$. We investigate equivalent condition for each case that $D_aF(x)=b$ has at least two solutions in $P_a$.
\begin{equation*}
\begin{array}{ll}
_cD_{\gamma-1}F(0)={}_cD_{\gamma-1}F(1)\ \Leftrightarrow\ c=\frac{\gamma(\gamma-2)}{\gamma-1}, &  
_cD_{\gamma-1}F(0)={}_cD_{\gamma-1}F(\gamma)\ \Leftrightarrow\ c=\frac{-\gamma}{(2\gamma-1)(\gamma-1)},\\
_cD_{\gamma-1}F(0)={}_cD_{\gamma-1}F(1-\gamma)\ \Leftrightarrow\ c=1, & 
_cD_{\gamma-1}F(0)={}_cD_{\gamma-1}F(2-\gamma)\ \Leftrightarrow\ c=\frac{\gamma-2}{\gamma(\gamma-1)},\\
_cD_{\gamma-1}F(1)={}_cD_{\gamma-1}F(\gamma)\ \Leftrightarrow\ c=-\frac{2\gamma}{2\gamma-1}, &  
_cD_{\gamma-1}F(1)={}_cD_{\gamma-1}F(1-\gamma)\ \Leftrightarrow\ c=\frac{\gamma(\gamma-1)}{2\gamma-1}, \\
_cD_{\gamma-1}F(1)={}_cD_{\gamma-1}F(2-\gamma)\ \Leftrightarrow\ c=\frac{\gamma-2}{2}, &  
_cD_{\gamma-1}F(\gamma)={}_cD_{\gamma-1}F(1-\gamma)\ \Leftrightarrow\ c=\frac{\gamma-1}{\gamma(2\gamma-1)},\\
_cD_{\gamma-1}F(\gamma)={}_cD_{\gamma-1}F(2-\gamma) \Leftrightarrow c=\frac{2-\gamma}{\gamma(2\gamma-1)},  &  
_cD_{\gamma-1}F(1-\gamma)={}_cD_{\gamma-1}F(2-\gamma) \Leftrightarrow c=-\frac{(\gamma-1)(\gamma-2)}{\gamma}.
\end{array}
\end{equation*}
\begin{itemize}
\item $_cD_{\gamma-1}F(0)={}_cD_{\gamma-1}F(1)={}_cD_{\gamma-1}F(\gamma)\ \Leftrightarrow\ c=\frac{\gamma(\gamma-2)}{\gamma-1}=\frac{-\gamma}{(2\gamma-1)(\gamma-1)}=-\frac{2\gamma}{2\gamma-1}$. Then we have $\gamma=\frac{3}{2}$ and then $c=-\frac{3}{2}$.
\item $_cD_{\gamma-1}F(0)={}_cD_{\gamma-1}F(1)={}_cD_{\gamma-1}F(1-\gamma)\ \Leftrightarrow\ c=\frac{\gamma(\gamma-2)}{\gamma-1}=\frac{\gamma(\gamma-1)}{2\gamma-1}=1$. Then we have $\gamma^2-3\gamma+1=0$, such $\gamma$ exists if and only if $(-3)^2-4=5$ is a square if and only if $p^n \equiv 0,\pm1 \pmod{5}$.
\item $_cD_{\gamma-1}F(0)={}_cD_{\gamma-1}F(1)={}_cD_{\gamma-1}F(2-\gamma)\ \Leftrightarrow\ c=\frac{\gamma(\gamma-2)}{\gamma-1}=\frac{\gamma-2}{\gamma(\gamma-1)}=\frac{\gamma-2}{2}$. Then we have $\gamma=-1$ and then $c=-\frac{3}{2}$.
\item $_cD_{\gamma-1}F(0)={}_cD_{\gamma-1}F(\gamma)={}_cD_{\gamma-1}F(1-\gamma)\ \Leftrightarrow\ c=\frac{-\gamma}{(2\gamma-1)(\gamma-1)}=\frac{\gamma-1}{\gamma(2\gamma-1)}=1$. Then we have $2\gamma^2-2\gamma+1=0$, such $\gamma$ exists if and only if $(-2)^2-8=-4$ is a square if and only if $p^n \equiv 1 \pmod{4}$.
\item $_cD_{\gamma-1}F(0)={}_cD_{\gamma-1}F(\gamma)={}_cD_{\gamma-1}F(2-\gamma)\ \Leftrightarrow\ c=\frac{-\gamma}{(2\gamma-1)(\gamma-1)}=\frac{\gamma-2}{\gamma(\gamma-1)}=\frac{2-\gamma}{\gamma(2\gamma-1)}$. Then we have $\gamma=\frac{2}{3}$ and then $c=6$.
\item $_cD_{\gamma-1}F(0)={}_cD_{\gamma-1}F(1-\gamma)={}_cD_{\gamma-1}F(2-\gamma)\ \Leftrightarrow\ c=1=\frac{\gamma-2}{\gamma(\gamma-1)}=-\frac{(\gamma-1)(\gamma-2)}{\gamma} $. Then we have $\gamma^2-2\gamma+2=0$, such $\gamma$ exists if and only if $(-2)^2-8=-4$ is a square if and only if $p^n \equiv 1 \pmod{4}$.
\item $_cD_{\gamma-1}F(1)={}_cD_{\gamma-1}F(\gamma)={}_cD_{\gamma-1}F(1-\gamma)\ \Leftrightarrow\ c=-\frac{2\gamma}{2\gamma-1}=\frac{\gamma(\gamma-1)}{2\gamma-1}=\frac{\gamma-1}{\gamma(2\gamma-1)}$. Then we have $\gamma=-1$ and then $c=-\frac{2}{3}$.
\item $_cD_{\gamma-1}F(1)={}_cD_{\gamma-1}F(\gamma)={}_cD_{\gamma-1}F(2-\gamma)\ \Leftrightarrow\ c=-\frac{2\gamma}{2\gamma-1}=\frac{\gamma-2}{2}=\frac{2-\gamma}{\gamma(2\gamma-1)}$. Then we have $2\gamma^2-\gamma+2=0$, such $\gamma$ exists if and only if $(-1)^2-4\cdot 2\cdot 2=-15$ is a square if and only if $p^n\equiv 0,1,2,4,6,8,9,10 \pmod{15}$.
\item $_cD_{\gamma-1}F(1)={}_cD_{\gamma-1}F(1-\gamma)={}_cD_{\gamma-1}F(2-\gamma)\ \Leftrightarrow\ c=\frac{\gamma(\gamma-1)}{2\gamma-1}=\frac{\gamma-2}{2}=-\frac{(\gamma-1)(\gamma-2)}{\gamma}$. Then we have $\gamma=\frac{2}{3}$ and then $c=-\frac{2}{3}$.
\item $_cD_{\gamma-1}F(\gamma)={}_cD_{\gamma-1}F(1-\gamma)={}_cD_{\gamma-1}F(2-\gamma)\ \Leftrightarrow\ c=\frac{\gamma-1}{\gamma(2\gamma-1)}=\frac{2-\gamma}{\gamma(2\gamma-1)}=-\frac{(\gamma-1)(\gamma-2)}{\gamma}$. Then we have $\gamma=\frac{3}{2}$ and then $c=\frac{1}{6}$.
\end{itemize}
From above we can see that $_cD_{\gamma-1}F(x)=b$ has at least four solutions only when $p=5$, $\gamma=-1$ and $c=1$. For example, if $_cD_{\gamma-1}F(0)={}_cD_{\gamma-1}F(1)={}_cD_{\gamma-1}F(\gamma)={}_cD_{\gamma-1}F(1-\gamma)$ then $c=-\frac{3}{2}=1=-\frac{2}{3}$ and $\gamma=\frac{3}{2}=-1$ and $\gamma^2-3\gamma+1=2\gamma^2-2\gamma+1=0$ can happen only when $p=5$. Hence we can see that all element in $P_a$ is a solution of $_cD_{\gamma-1}F(x)=b$ if and only if $p=5$, $\gamma=-1$ and $c=1$, and note that $a=3$ and $b=2$ in this case. By Lemma \ref{cdu_symm_lemma}, $_1D_{2}F(x)=3$ also have five solutions in $P_a$.
\end{appendix}

\end{document}